\documentclass{article}
\usepackage[utf8]{inputenc}

\usepackage[margin=1in]{geometry}
\usepackage{setspace}
\usepackage{graphicx}
\usepackage[dvipsnames]{xcolor}
\usepackage{fancyhdr} 
\usepackage{amsmath, amsfonts, amsthm}
\usepackage{tikz}
\usepackage{authblk}

\usepackage{hyperref}
\usepackage{bbm}
\usepackage{nth}
\usepackage{dsfont}
\usepackage{subfig}
\usepackage{algorithm2e} 
\usepackage[noabbrev,nameinlink,capitalize,sort]{cleveref}
\usepackage{booktabs}

\providecommand{\keywords}[1]
{
  {\small	
  \textbf{\textit{Keywords---}} #1}
}

\newcommand\numberthis{\addtocounter{equation}{1}\tag{\theequation}}

\usepackage{xpatch}
\makeatletter
\AtBeginDocument{\xpatchcmd{\cref@thmoptarg}{\thm@headpunct{.}}{\thm@headpunct{}}{}{}}
\AtBeginDocument{\xpatchcmd{\cref@thmnoarg}{\thm@headpunct{.}}{\thm@headpunct{}}{}{}}
\makeatother

\SetKwComment{Comment}{/* }{ */}
\RestyleAlgo{ruled}
\SetKwInput{KwInput}{Input}   
\SetKwInput{KwOutput}{Output}
\SetKwInput{Init}{init}
\SetKwInput{Return}{Return}
\SetKw{KwBy}{by}
\newcommand{\algindent}{{\mbox{}\phantom{\textbf{aaa} \itshape(}}}

\DeclareMathOperator*{\essinf}{ess\,inf}
\DeclareMathOperator*{\esssup}{ess\,sup}

\newcommand{\KL}[1]{D_{KL}(#1 \, \Vert \, \P)}

\newtheorem{theorem}{Theorem}[section]
\newtheorem{lemma}[theorem]{Lemma}
\newtheorem{definition}[theorem]{Definition}
\newtheorem{proposition}[theorem]{Proposition}
\newtheorem{corollary}[theorem]{Corollary}
\newtheorem{example}[theorem]{Example}

\newtheorem{optimization}[theorem]{Optimization Problem}

\newcommand{\boldeta}{{\boldsymbol{\eta}}}

\newcommand{\E}{{\mathbb{E}}}
\newcommand{\Q}{{\mathbb{Q}}}

\renewcommand{\P}{{\mathbb{P}}}

\newcommand{\Id}{{\mathds{1}}}
\newcommand{\R}{{\mathds{R}}}

\newcommand{\Qset}{{\mathcal{Q}}}

\newcommand{\F}{{\mathcal{F}}}
\newcommand{\B}{{\mathcal{B}}}

\renewcommand{\L}{{\mathcal{L}}}

\newcommand{\VaR}{{\text{VaR}}}

\newcommand{\dQP}{\frac{d\Q}{d\P}}

\hypersetup{
  colorlinks   = true, 
  urlcolor     = blue, 
  linkcolor    = red, 
  citecolor   = black 
}

\allowdisplaybreaks 

\usepackage{todonotes}

\usepackage[natbib=true, style=authoryear]{biblatex}
\bibliography{refs.bib}
\renewbibmacro{in:}{}  

\title{Stressing Dynamic Loss Models}
\author[1,a]{Emma Kroell}
\author[1,b]{Silvana M. Pesenti}
\author[1,c]{Sebastian Jaimungal}

\affil[1]{Department of Statistical Sciences, University of Toronto}
\affil[a]{emma.kroell@mail.utoronto.ca}
\affil[b]{silvana.pesenti@utoronto.ca}
\affil[c]{sebastian.jaimungal@utoronto.ca}
\date{\today}

\begin{document}

\maketitle

\begin{abstract}
Stress testing, and in particular, reverse stress testing, is a prominent exercise in risk management practice. Reverse stress testing, in contrast to (forward) stress testing, aims to find an alternative but plausible model such that under that alternative model, specific adverse stresses (i.e. constraints) are satisfied. Here, we propose a reverse stress testing framework for dynamic models. Specifically, we consider a compound Poisson process over a finite time horizon and stresses composed of expected values of functions applied to the process at the terminal time. We then define the stressed model as the probability measure under which the process satisfies the constraints and which minimizes the Kullback-Leibler divergence to the reference compound Poisson model.

We solve this optimization problem, prove existence and uniqueness of the stressed probability measure, and provide a characterization of the Radon-Nikodym derivative from the reference model to the stressed model. We find that under the stressed measure, the intensity and the severity distribution of the process depend on time and state, and hence the stressed model is not a compound Poisson process. We illustrate the dynamic stress testing by considering stresses on VaR and both VaR and CVaR jointly and provide illustrations of how the stochastic process is altered under these stresses. We generalize the framework to multivariate compound Poisson processes and stresses at times other than the terminal time. We illustrate the applicability of our framework by considering ``what if'' scenarios, where we answer the question: What is the severity of a stress on a portfolio component at an earlier time such that the aggregate portfolio exceeds a risk threshold at the terminal time? Furthermore, for general constraints, we propose an algorithm to simulate sample paths under the stressed measure, thus allowing to compare the effects of stresses on the dynamics of the process.
\end{abstract}

\keywords{Reverse Stress Testing, Compound Poisson Processes, KL divergence, Value-at-Risk, Conditional Value-at-Risk}

\onehalfspacing

\section{Introduction}

Since the 2007-08 financial crisis, there has been an increased focus in the insurance and financial sector on modelling and assessing rare but extreme events. Stress testing, which examines a model or system under extreme but realistic scenarios, has become a key aspect of risk management at financial institutions, as well as a vital tool for regulators \citep{BCBS2018}. In light of the regulatory requirements on stress testing, a large portion of the literature on stress testing is concerned with constructing plausible but \emph{adverse scenarios} that result in a substantial stress on the model, see e.g., \textcite{Berkowitz1999Risks} and \textcite{Meucci2008Risks}. 
For example, \textcite{Mcneil2012IME} develop, in the context of solvency requirements, a stress testing methodology that provides multivariate adverse scenarios via the notion of a ``least solvency likely event''. \textcite{Glasserman2015QF} construct non-parametrically the most likely scenarios that lead to losses exceeding a prespecified risk threshold. Connecting to systemic risk, \textcite{brechmann2013IME} utilize vine copulas to perform stress tests to investigate systemic risk and contagion in financial networks. 
In an application to credit risk, \textcite{Breuer2012JBF} consider multi-period adverse scenarios whose plausibility is measured via the Mahalanobis distance of risk factors changes. \textcite{Breuer2013JBF} highlight the difficulty of choosing realistic but dangerous scenarios, which are essential for reliable stress tests. We also refer to \textcite{StressTest2013book} for a treatment of stress testing from a regulatory point of view.

An alternative approach to stress testing -- sometimes termed \emph{reverse stress testing} -- is to first define an  adverse or stressed state of the world, and then characterize the probability distribution leading to it, instead of defining a stress and searching for the corresponding adverse scenario. Fundamental to this framework is (a) the specification of \emph{stresses} or probabilistic constraints and (b) solving an optimization problem to find the probability measure under which the model satisfies the stresses and is ``close'' to a reference probability measure. One key advantage of reverse stress testing is that it provides the entire stressed probabilistic model, allowing for a comparison of the stressed model with the reference model. In a discrete and static setting, \textcite{Cambou2017} use the $f$-divergence to incorporate constraints on events, which they term ``views'', while \textcite{Makam2021} consider the $\chi^2$-divergence and a constraint on the expected value of a risk factor. For general but static models, \textcite{pesenti2019reverse} and \textcite{Pesenti2022Risks} consider stresses defined via changes in risk measures such as Value-at-Risk (VaR), distortion risk measures, and expected utilities, using the Kullback-Leibler (KL) and the Wasserstein distance, respectively.

While all the aforementioned works consider static settings, and thus do not capture the potential impact of a stress cascading through time, \textcite{Breuer2012JBF} and \textcite{Bellotti2013IJF} highlight the need for understanding the effect of stress testing dynamic models. Indeed, in a plethora of settings the time-development of an insurance or financial portfolios' profit and loss (P\&L) is of interest, and thus also its evolution under adverse stresses. Examples in actuarial science include claims reserving (\cite{Merz2008CAS}), variable annuities which rely on accurate dynamic mortality models (\cite{Bacinello2011IME}), optimal pension accumulation and withdrawal (\cite{Gerrard2004IME}), and valuation and hedging of insurance risks (\cite{Dahl2006IME}). In all these examples, the key interest lies in the dynamics of the (stressed) P\&L. Moreover, quantifying the stressed dynamics of extreme losses allows for adequate risk assessment and management interventions, see e.g., \textcite{Merz2008CAS} for a detailed discussion. 

The literature on stress testing in a dynamic 
 setting, however, is limited with \textcite{Breuer2012JBF} and \textcite{Bellotti2013IJF}, who construct
 adverse scenarios in a multi-period and thus discrete-time setting within a credit risk context. As such, a comprehensive treatment of stress testing for dynamic loss models in continuous-time is missing. We fill this gap by developing a (reverse) stress testing approach for compound Poisson processes in a continuous-time setting, which allows for multiple stresses on portfolio components, stresses at terminal and earlier time points, and provide an algorithm that allows to simulate under the stressed probability measure. While the starting point is a compound Poisson process, we allow the stressed dynamics to fall in the class of stochastic processes with state dependent intensity and state dependent severity distributions. A key advantage of our framework is that it explicitly quantifies the dynamics of the stressed stochastic process over the entire time horizon. Furthermore, the developed algorithm allows for efficient simulation of the stressed processes and thus provides a methodology to estimate, e.g., future losses and risk measures, under the stressed scenarios. 

In this work, we contribute to the stress testing literature by generalizing the reverse stress testing approach to a dynamic setting, via stressing a stochastic process at a fixed point in time and characterizing the dynamics of the processes under the stressed probability measure. In particular, we consider an insurance portfolio over a finite time horizon, which under the reference probability measure is modelled by a compound Poisson process. We then are interested in how the dynamics of the portfolio need to be changed so that, under the altered dynamics, stress(es) at the terminal time are attained. The stresses we consider are expectations of functions applied to the portfolio at the terminal time. This includes stresses on risk measures such as Value-at-Risk (VaR), VaR and Conditional Value-at-Risk (CVaR), and expected utility constraints. To find the \emph{stressed probability measure}, we seek over equivalent probability measures on the path-space of stochastic processes the one under which the process attains the stresses and which minimizes the KL divergence to the reference measure. We solve the optimization problem, characterize the unique stressed probability measure and the stressed dynamics of the stochastic process. We find that under the stressed probability measure, the intensity and the severity distribution of the process both depend on state and time. Thus, in the spirit of reverse stress testing, our analysis informs how the effects of a stress on a Poisson process results in a marked point process with time and state dependence in both intensity and severity distribution. 
Differently from the static setting of reverse stress testing or the development of adverse scenarios, we require and are interested in a full characterization of the dynamics of the stochastic process under the stressed probability measure that lead to the constraint. Therefore, we also propose an algorithm to simulate sample paths under the stressed probability measures, and thus are able to illustrate the effects of a stress on the entire process. 

Relevant for risk management purposes, we study a stress on VaR and a stress on VaR and CVaR jointly and provide illustrations of how the stochastic process is altered under these stresses. We further consider generalizations to multivariate compound Poisson processes and stresses at times earlier than the terminal time. For stresses on a sub-portfolio, we observe that the dependence between the sub-portfolios may amplify the stress. We illustrate the applicability of our framework by considering ``what if'' scenarios, where we examine what severity of stress on a portfolio component at an earlier time is needed such that the aggregate portfolio exceeds a risk threshold at the terminal time.

The paper is structured as follows. \cref{sec:main} introduces and solves the optimization problem, including an alternative representation of the change of measure and the stressed dynamics of the process. In \cref{sec:applications}, we discuss two applications: a stress on VaR of the process at the terminal time and a stress on VaR and CVaR jointly. In \cref{sec:extensions}, we consider two extensions of our framework: first to allow for constraints that are not at the terminal time, and second to multivariate compound Poisson processes, where we examine the cascading effect a stress on one component of the process has on the other components. The algorithm for simulating the process under the stressed measure, which is used throughout the paper, is discussed in
\cref{sec:numerics}. The code is available at \href{https://github.com/emmakroell/stressing-dynamic-loss-models}{https://github.com/emmakroell/stressing-dynamic-loss-models}.

\section{Minimally Perturbed Compound Poisson Process}
\label{sec:main}
Let a filtered probability space $(\Omega, \F, \{\F_t\}_{t\in[0,T]}, \P)$ be given and let $\B(\cdot)$ denote the Borel $\sigma$-algebra of a set. We consider a one-dimensional jump process $X := (X_t)_{t \in [0,T]}$ that satisfies the stochastic differential equation (SDE) under $\P$
\begin{equation*}
    dX_t = \int_{\R} x\, \mu(dx,dt),
\end{equation*}
where $\mu$ is a Poisson random measure. We assume for every $t\in [0,T]$ and every $A\in\B(\R)\otimes\B([0,t])$ that $\mu(\cdot,A)$ is $\F_t$-measurable and that $\mu(\cdot,\R \times (t,\infty))$ is independent of $\F_t$. Thus, the random measure $\mu$ is $\F$-adapted and its future increments are independent of the past filtration. We further assume that $X$ has mean measure (also known as the compensating measure) given by
\begin{equation*}
    \nu(dx,dt) = \kappa \, G(dx) dt,
\end{equation*} 
thus $X$ is a compound Poisson process with intensity $\kappa>0$ and severity or jump size distribution $G$. We denote by  $\tilde \mu (dy, dt) := \mu (dy, dt) - \nu (dy, dt)$ the compensated measure.

Throughout, we call $\P$ the \textit{reference probability measure} and denote by $\Q$ alternative probability measures on the same filtered probability space. For notational convenience, we write $\E^\Q[\cdot]$ for the expected value under $\Q$ and set $\E[\cdot] := \E^\P[\cdot ]$. To quantify discrepancies between probability measures, we use the Kullback-Leibler (KL) divergence \citep{kullback1951information}. The KL divergence of $\Q$ with respect to $\P$ (also known as the relative entropy) is defined as
\begin{equation*}
     D_{KL}(\mathbb{Q} \, \Vert \, \mathbb{P} ) = \begin{cases}
     \mathbb{E} \left[ \frac{d\mathbb{Q}}{d\mathbb{P}} \; \log \left( \frac{d\mathbb{Q}}{d\mathbb{P}} \right) \right]  &  \text{if } \Q \ll \P  \, ,\\
     \infty & \text{otherwise}\,,
     \end{cases} 
\end{equation*}
where we use the convention $0 \log 0 = 0$ and $\dQP$ denotes the Radon-Nikodym (RN) derivative of $\Q$ with respect to $\P$. For an overview of the properties of KL divergence, see, e.g., \textcite{van2014renyi}.
 
With this notation at hand, we first motivate the optimization problem which is central to this work. Consider an insurance portfolio modelled by a compound Poisson process $X$ under the reference measure $\P$. A modeller wishes to examine how the process is perturbed under \textit{stresses} on its terminal value, i.e., on $X_T$. Here, we consider stresses of the form $\E^\Q\left[f_i(X_T)\right]=c_i$, for functions $f_i \colon \R \to \R$ and constants $c_i\in \R$, $i = 1,\ldots, n$. Examples of stresses include expected utility constraints and risk measure constraints such as Value-at-Risk (VaR) and Conditional Value-at-Risk (CVaR). For a specified stress, the modeller then seeks the most ``plausible'' probability measure under which the stochastic process achieves the stresses. 

Specifically, we consider equivalent probability measures characterized by Girsanov's theorem \citep{appliedstochjump}, i.e. those with RN derivatives of the form
\begin{equation*}
\frac{d\Q}{d\P} =\mathfrak{E}\left(\int_0^T \int_\R \left[ h_t(y) - 1 \right] \, \tilde \mu(dy,dt)\right),
\end{equation*}
where $\mathfrak{E}(\cdot)$ denotes the stochastic exponential and where $h:=(h_t)_{t\in [0,T]}$ is a predictable, non-negative random field on $[0,T]$. To assure that $\frac{d\Q}{d\P} $ is indeed a martingale, we assume that $h$ satisfies the following condition.

\begin{definition}
The random field $h:=(h_t)_{t\in [0,T]}$ satisfies the Novikov condition on the interval $[0,T]$ if
\begin{equation*}
\label{eq:novikov}
   \E \left[  \exp \left( \int_0^\tau \int_\R (1 - h_t(y))^2 \, \mu(dy,dt) \right) \right] < \infty . 
\end{equation*}
\end{definition}

We seek over probability measures of this form the one with minimal KL divergence to $\P$ under which the process attains the constraints. Mathematically, we consider the following optimization problem.
\begin{optimization}
\label{opti:main}
Let $f_i \colon \R \to \R$ and $c_i \in \R$ for $i \in [n]$, where $[n] := \{1, \ldots, n \}$, and consider
\begin{equation*}
    \inf_{\Q\in\Qset} D_{KL}(\mathbb{Q} \, \Vert \, \mathbb{P} )
    \quad
    \text{s.t.}
    \quad
    \E^\Q\left[f_i(X_T)\right]=c_i, \quad i \in [n]\,,
\end{equation*}
where $\Qset$ is the class of equivalent probability measures induced by Girsanov's theorem, i.e.,
\begin{equation*}
\mathcal{Q}:=\left\{\Q_h\; \Big| \; \frac{d\Q_h}{d\P}=\mathfrak{E}\left(\int_0^T \int_\R \left[ h_t(y) - 1 \right] \, \tilde \mu(dy,dt)\right) \right\} ,
\end{equation*}
where $h_t$ is a predictable, non-negative random field satisfying Novikov's condition on the interval $[0,T]$.
\end{optimization}

Note that any probability measure $\Q_h \in \mathcal{Q}$ is uniquely characterized by a random field $h$. We refer to any predictable, non-negative random field that satisfies the Novikov condition on $[0,T]$ as a \textit{Girsanov kernel}, denoted by $h$. We further note that $\Qset$ is convex since a convex combination of random fields is a random field and any convex combination of random fields satisfying the Novikov condition also satisfies the Novikov condition.
In general, a change of measure induces a modified intensity and severity distribution that may be time and space dependent. Proposition \ref{prop:kappa_G} provides the exact characterization of the changes induced by the optimal measure change.

The next result provides the Girsanov kernel which characterizes the probability measure attaining the infimum in \cref{opti:main}.  

\begin{theorem}
\label{thm:h}
If a solution to \cref{opti:main} exists, it is given by $\Q_h$ with $h$ characterized by $h_t(y)=h^\boldeta(t, X_{t^{-}}, y)$, where
\begin{align}
h^\boldeta(t,x,y) := \frac{\E_{t,x+y}\left[\exp \left( -\sum_{i=1}^n\eta_i\,f_i(X_T) \right)\right]}{\E_{t,x}\left[\exp \left( -\sum_{i=1}^n\eta_i\,f_i(X_T) \right)\right]},
\label{eq:h}
\end{align}
$\boldsymbol{\eta}=(\eta_1, \ldots , \eta_n) \in \R^n$ are Lagrange multipliers such that the constraints hold, and $\E_{t,x}[\cdot]$ denotes the $\P$-expectation conditional on the event that the process $X$ at time $t^-$ is equal to $x$, i.e., $X_{t^-}=x$. The solution is unique.
\end{theorem}
 
\begin{proof}
Let $\Q_h \in \Qset$, then the RN derivative $\frac{d\Q_h}{d\P}$ has the representation
\begin{equation*}
    \frac{d\Q_h}{d\P} = \mathfrak{E}\left(\int_0^T \int_\R \left[ h_t(y) - 1 \right] \, \tilde \mu(dy,dt)\right) = \exp \left(- \int_0^T\int_\R \left(h_t(y)-1\right)\,\nu(dy,dt)
    + \int_0^T\int_\R\log h_t(y)\;\mu(dy,dt) \right)
\end{equation*}
and the KL divergence of $\Q_h$ with respect to $\P$ becomes
\begin{align*}
\E\left[\frac{d\Q_h}{d\P} \,\log \left( \frac{d\Q_h}{d\P} \right) \right] &= \E^{\Q_h}\left[\left(- \int_0^T\int_\R \left(h_t(y)-1\right)\,\nu(dy,dt)
    + \int_0^T\int_\R\log h_t(y)\;\mu(dy,dt) \right) \right] \\
&= \E^{\Q_h}\left[ \int_0^T\int_\R \big(1-(1-\log h_t(y))\;h_t(y) \big)\,\nu(dy,dt)  \right] \\
&= \E^{\Q_h}\left[ \int_0^T\int_\R \big(1-\left(1-\log h_t(y)\right)\;h_t(y) \big)\,\kappa \, G(dy) dt  \right] .
\end{align*}

Let $\boldeta = (\eta_1, \ldots, \eta_n)$ be Lagrange multipliers. The Lagrangian associated with this optimization problem can be written as
\begin{equation*}
    \inf_{\Q_h\in\Qset} \E^{\Q_h}\left[\log\frac{d\Q_h}{d\P} + \sum_{i=1}^n\eta_i \left(f_i(X_T)-c_i\right)\right].
\end{equation*}
The time $t$ version of the associated value function for given Lagrange multipliers is defined as
\begin{equation*}
    J^\boldeta(t,x) := \inf_{\Q_h\in\Qset} \E^{\Q_h}_{t,x}\left[ \int_t^T\int_\R \big(1-(1-\log h_t(y))\;h_t(y) \big)\,\kappa \, G(dy) dt  + \sum_{i=1}^n\eta_i \left(f_i(X_T)-c_i\right)\right] ,
\end{equation*}
where $\E^{\Q_h}_{t,x}[\cdot]$ denotes the $\Q_h$-expectation given that the process $X$ at time $t^-$ is equal to $x$, i.e., $X_{t^-}=x$.

Under the assumption that $h_t$ is Markov, so that we may write $h_t(y)=h(t,X_{t^-},y)$ for some function $h:\R_+^3\mapsto \R$, $\R_+ := [0, \infty)$, the dynamic programming principle implies that the value function should satisfy the Hamilton-Jacobi-Bellman (HJB) equation:
\begin{subequations}
\begin{align}
    \partial_t J^\boldeta(t,x)+ \inf_{h} \left\{ \L^{h} J^\boldeta(t,x) + \int_\R \big(1- (1-\log h(t,x,y) ) \,h(t,x,y) \big) \kappa \, G(dy)  \right\} &=0 ,
    \label{eqn:HJBa}
    \\
    J^\boldeta(T,x) &= \sum_{i=1}^n\eta_i \left(f_i(x)-c_i\right),
\end{align}
\label{eqn:HJB}
\end{subequations}
where the linear operator $\L^{h}$ is the $\Q_{h}$-generator of $X$, and acts on functions as follows:
\begin{equation*}
    \L^h J^\boldeta(t,x) = \int_\R \left[J^\boldeta(t,x+y)-J^\boldeta(t,x)\right]\,h(t,x,y)\,\kappa \, G(dy) .
\end{equation*}

Applying the first order conditions to the $\inf$ term in \eqref{eqn:HJBa}, for fixed Lagrange multipliers, we obtain  the optimal control $h^\boldeta$ in feedback form:
\begin{equation}
  h^\boldeta(t,x,y) = e^{- \Delta_y J^\boldeta(t,x)}, 
  \label{eqn:opt-h}
\end{equation}
where $\Delta_y J^\boldeta(t,x) := J^\boldeta(t,x+y)-J^\boldeta(t,x)$.
Inserting the feedback form of the control back into \eqref{eqn:HJBa}, we obtain
\begin{equation}
    \partial_t J^\boldeta(t,x) + \int_\R \big(1- e^{-\Delta_y J^\boldeta(t,x)}\big) \kappa\,G(dy) = 0.
    \label{eqn:HJB-at-optimal}
\end{equation}

Next, we construct an explicit solution to \eqref{eqn:HJB-at-optimal} by using a Cole-Hopf change of variables $J^\boldeta(t,x)=-\log \omega^\boldeta(t,x)$. In this case,  \eqref{eqn:HJB-at-optimal} reduces to
\begin{equation*}
-\frac{\partial_t\omega^\boldeta(t,x)}{\omega^\boldeta(t,x)} + \int_\R\left(1-\frac{\omega^\boldeta(t,x+y)}{\omega^\boldeta(t,x)}\right) \kappa\,G(dy)= 0\,.
\end{equation*}
Multiplying through by $-\omega^\boldeta(t,x)$, we obtain the linear PDE
\begin{equation}
\label{eqn:HJB-linearized}
    \partial_t\omega^\boldeta(t,x) + \L \omega^\boldeta(t,x) = 0, \qquad \text{s.t. }\quad 
    \omega^\boldeta(T,x) = \exp\left( - \sum_{i=1}^n\eta_i \left(f_i(x)-c_i\right) \right)\, ,
\end{equation}
where the operator $\L$ is  the $\P$-generator of the process $X$ and acts on functions as follows
\begin{equation*}
    \L \omega^\boldeta(t,x) = \int_\R\left(\omega^\boldeta(t,x+y)-\omega^\boldeta(t,x)\right) \kappa(t,x)\,G_t(dy)\,.
\end{equation*}

This is a Cauchy linear parabolic PDE, which, via the Feynman-Kac representation \citep{pham2009}, admits the solution
\begin{equation*}
    \omega^\boldeta(t,x) = \E_{t,x}\left[ \exp\left( - \sum_{i=1}^n\eta_i \left(f_i(X_T)-c_i\right) \right) \right].
\end{equation*}
Inserting this representation into the feedback form of the optimal control \eqref{eqn:opt-h}, we have 
\begin{align*}
h^\boldeta(t,x,y) = \frac{\omega^\boldeta(t,x+y)}{\omega^\boldeta(t,x)} =\frac{\E_{t,x+y}\left[\exp \left( -\sum_{i=1}^n\eta_i\,f_i(X_T) \right)\right]}{\E_{t,x}\left[\exp \left( -\sum_{i=1}^n\eta_i\,f_i(X_T) \right)\right]}.
\end{align*}
The uniqueness of the solution follows from the convexity of $\Qset$ and the KL divergence.
\end{proof}
For simplicity, we write $\Q^{\boldeta} := \Q_{h^\boldeta}$ whenever the probability measure is induced by $h^\boldeta$ given in \eqref{eq:h}. Conditions under which the Lagrange multipliers exist are discussed in \cref{prop:optim_eta_gen}. Next, we derive an alternative representation of the RN derivative $\Q_h$ associated with $h(t,x,y)$, when $h(t,x,y)$ has the specific form given in \eqref{eq:h}. For this we first state a general result.

\begin{theorem}
\label{thm:special-stoch-exp}
Define the functions
\begin{align*}
    \ell:\R_+\times \R^2\mapsto \R_{>0}, \quad
    \text{such that} \quad  \ell(t,x,y)
    &:= \frac{\theta(t,x+y)}{\theta(t,x)},
    \\
    \theta:\R_+\times \R\mapsto \R, \quad 
    \text{such that} \quad \theta(t,x) &:=\E_{t,x}[\mathfrak{Z}(X_T)],
\end{align*}
for some function $\mathfrak{Z}:\R\mapsto \R_{>0}$ with $\E[\mathfrak{Z}(X_T)]<\infty$. Then, it holds $\P$-a.s. that
\begin{equation}
    \exp\left\{\int_0^T\int_\R \log \ell(t,X_{t^-},y)\,\mu(dy,dt) - \int_0^T\int_\R \left(\ell(t,X_{t^-},y)-1\right)\,\nu(dy,dt)\right\}
    = \frac{\mathfrak{Z}(X_T)}{\E\left[\,\mathfrak{Z}(X_T)\,\right]}.
    \label{eqn:stoch-exp-of-ell}
\end{equation}
\end{theorem}
\begin{proof}
Define the stochastic process $Z=(Z_t)_{t\in[0,T]}$, such that 
\begin{equation*}
    Z_t = \exp\left\{\int_0^t\int_\R \log \ell(s,X_{s^-},y)\,\mu(dy,ds) - \int_0^t\int_\R \left(\ell(s,X_{s^-},y)-1\right)\,\nu(dy,ds)\right\}.
\end{equation*}
Clearly, $Z_T$ coincides with the left-hand side of \eqref{eqn:stoch-exp-of-ell}. By It\^{o}'s lemma we have that
\begin{equation}
    dZ_t = Z_{t^-}\int_\R \left(\ell(t,X_{t^-},y)-1\right)[\mu(dy,dt)-\nu(dy,dt)],
    \label{eqn:sde-for-Z}
\end{equation}
and therefore $Z$ is a local martingale. Moreover, due to the integrability and support of $\mathfrak{Z}$, $\ell$ is bounded, and therefore, $Z$ is a true martingale. 

Next, define the stochastic process $\Theta:=(\Theta_t)_{t\in[0,T]}$, where $\Theta_t:=\theta(t,X_{t^-})$. As $X$ is a Markov process, $\Theta_t=\E_t[\mathfrak{Z}(X_T)]$ and by integrability of $\mathfrak{Z}$, $\Theta$ is also a martingale, hence, 
\begin{equation*}
    d\Theta_t = \int_\R \left(\theta(t,X_{t^-}+y)-\theta(t,X_{t^-})\right)[\mu(dy,dt)-\nu(dy,dt)].
\end{equation*}
We aim to show that $Z_t/\Theta_t$ is $\P$-a.s constant for all $t \in [0,T]$. To this end, from It\^{o}'s lemma we have that
\begin{equation*}
    d\left(\frac{1}{\Theta_t}\right) =
  -\frac{\partial_t\theta(t,X_{t^-})}{\theta^2(t,X_{t^-})}\, dt
   +\int_\R \left(\frac{1}{\theta(t,X_{t^-}+y)}-\frac{1}{\theta(t,X_{t^-})}\right)\mu(dy,dt).
\end{equation*}
Combining the above with \eqref{eqn:sde-for-Z}, we have
\begin{align*}
    d \left( \frac{Z_t}{\Theta_t}\right)
    &=  \frac{d Z_t}{\Theta_{t^-}} +   Z_{t^-} d \left( \frac{1}{\Theta_t}\right)  + d\left[Z,\frac1\Theta\right]_t
    \\
    \begin{split}
    &= \frac{Z_{t^-}}{\Theta_{t^-}}
    \int_\R \left(\ell(t,X_{t^-},y)-1\right)[\mu(dy,dt)-\nu(dy,dt)]
    \\
    &\quad -Z_{t^-} \frac{\partial_t\theta(t,X_{t^-})}{\theta^2(t,X_{t^-})}\,dt +Z_{t^-}\int_\R \left(\frac{1}{\theta(t,X_{t^-}+y)}-\frac{1}{\theta(t,X_{t^-})}\right)\mu(dy,dt)
    \\
    &\quad
    + Z_{t^-} \int_\R
    \left(\ell(t,X_{t^-},y)-1\right)
    \left(\frac{1}{\theta(t,X_{t^-}+y)}-\frac{1}{\theta(t,X_{t^-})}\right)\mu(dy,dt) \, .
    \end{split}
\end{align*}
As $\Theta_t=\theta(t,X_{t^-})$ is a martingale, we also have that
\begin{equation*}
    \partial_t \theta(t,x) +
    \int_\R\left(\theta(t,x+y)-\theta(t,x)\right)\,\kappa\,G(dy) = 0.
\end{equation*}
Using this identity and substituting the expression for $\ell$ in terms of $\theta$, yields
\begin{align*}
    \begin{split}
    d \left( \frac{Z_t}{\Theta_t}\right)
    &=
    \frac{Z_{t^-}}{\theta(t,X_{t^-})}
    \int_\R \left(
    \frac{\theta(t,X_{t^-}+y)}{\theta(t,X_{t^-})}-1\right)[\mu(dy,dt)-\nu(dy,dt)]
    \\
    &\quad
     + \frac{Z_{t^-}}{\theta^2(t,X_{t^-})}
     \int_\R\left(\theta(t,X_{t^-}+y)-\theta(t,X_{t^-})\right)\,\nu(dy,dt)
    \\
    &\quad +Z_{t^-}\int_\R \left(\frac{1}{\theta(t,X_{t^-}+y)}-\frac{1}{\theta(t,X_{t^-})}\right)\mu(dy,dt)
    \\
    &\quad + Z_{t^-} \int_\R
    \left(\frac{\theta(t,X_{t^-}+y)}{\theta(t,X_{t^-})}-1\right)
    \left(\frac{1}{\theta(t,X_{t^-}+y)}-\frac{1}{\theta(t,X_{t^-})}\right)\mu(dy,dt),
    \end{split}
\end{align*}
which, after inspecting the various terms, vanishes. Therefore, $Z_t = r \, \Theta_t$ for some $r \in\R$ and  all $t \in [0,T]$. As $Z$ is a martingale, $\E[Z_T]=Z_0=1$, moreover, $Z_T = r\,
\Theta_T=r \,\mathfrak{Z}(X_T)$. Therefore, $r=1/\E[\mathfrak{Z}(X_T)]$, which proves the result.
\end{proof}

With the above result we obtain a succinct representation of the RN derivative associated with the Girsanov kernel $h^\boldeta(t,X_{t-},y)$ from \cref{thm:h}, which is in the spirit of \textcite{Csiszar1975}.
\begin{corollary}
\label{thm:optim_RN}
For Lagrange multipliers $\boldeta = (\eta_1, \ldots,\eta_n)$ such that $\E[\exp \left( -\sum_{i=1}^n\eta_i\,f_i(X_T) \right)]<\infty$ and the Girsanov kernel $h^\boldeta$ given in \cref{thm:h}, the measure $\Q^\boldeta$ induced by $h^\boldeta$ has RN derivative 
\begin{equation*}
    \frac{d\Q^\boldeta}{d\P} = \frac{\exp \left( -\sum_{i=1}^n\eta_i\,f_i(X_T) \right)}{\E\left[\exp \left( -\sum_{i=1}^n\eta_i\,f_i(X_T) \right)\right]}
    \label{eq:optim_RN}.
\end{equation*}
\end{corollary}
\begin{proof}
Since
\begin{equation*}
\frac{d\Q^\boldeta}{d\P}=\exp \left(- \int_0^T\int_\R \left(h^\boldeta(t,X_{t-},y)-1\right)\,\nu(dy,dt)
    + \int_0^T\int_\R\log h^\boldeta(t,X_{t-},y)\;\mu(dy,dt) \right),
\end{equation*}
the result is an application of \cref{thm:special-stoch-exp} to the specific form of $h^\boldeta$ in \cref{thm:h}.
\end{proof}

Finally, we provide the existence of a solution to \cref{opti:main} by discussing the existence of the Lagrange multipliers. In particular, we show that the Lagrange multipliers are solutions of a (possibly non-linear) system of equations. 
\begin{proposition}
\label{prop:optim_eta_gen}
If there exists $\boldeta^*=(\eta^*_1, \ldots , \eta^*_n)$ such that
\begin{equation}
0 =  \E \left[ \exp \left( -\sum_{j=1}^n\eta^*_j\,f_j(X_T) \right)\left( f_i(X_T)-c_i \right)\right], \quad i \in [n],
\label{eq:find_eta}
\end{equation}
then \cref{opti:main} has a unique solution $\Q^* := \Q^{\boldeta^*}$ characterized by the Girsanov kernel $h^*_t(y) := h^{\boldeta^*}(t,X_{t^-},y)$.
\end{proposition}

\begin{proof}
For $\Q_h \in \Qset$ and $i \in [n]$, we rewrite the constraint(s) from \cref{opti:main} using the RN derivative from \cref{thm:optim_RN} to obtain
\begin{align*}
    0 = \E^\Q \left[ f_i(X_T) - c_i \right] = \E \left[ \frac{\exp \left( -\sum_{j=1}^n\eta_j\,f_j(X_T) \right)}{\E\left[\exp \left( -\sum_{j=1}^n\eta_j\,f_j(X_T) \right)\right]} \left( f_i(X_T)-c_i \right)\right] .
\end{align*}
Cancellation gives \eqref{eq:find_eta}.
\end{proof}

We refer to the solution of \cref{opti:main}, denoted by $\Q^*$, as the \textit{stressed probability measure}. Depending on the constraints, explicit expressions for $\boldeta^*$ may be available. We observe that whether Equation \eqref{eq:find_eta} has a solution depends on the choice of $f_i$ and $c_i$, $i \in [n]$. One necessary condition for Equation \eqref{eq:find_eta} to have a solution is that each $c_i$ must be in the support of $f_i(X_T)$. This requirement means that the value of each constraint $c_i$ is such it is attainable by $f_i(X_T)$ in some state of the world.

For the case of a single constraint, we can specify additional conditions on $f(X_T)$ and $c$ such that a solution to \cref{eq:find_eta} exists.
\begin{proposition}\label{prop:c-range}
    If the cumulant generating function of $f(X_T)$, $K_{f(X_T)}(s) := \log \E \left[ e^{sf(X_T)} \right]$, satisfies $K_{f(X_T)}(a), \, K_{f(X_T)}(b) < \infty$ for some $a,b \in \R$ with $a < b$, then for all 
    \begin{equation}
       c \in \left( \frac{\E\left[ e^{af(X_T)} f(X_T)\right]}{\E\left[ e^{af(X_T)}\right]},\frac{\E\left[ e^{bf(X_T)} f(X_T)\right]}{\E\left[ e^{bf(X_T)}\right]} \right) ,
       \label{eq:c-interval}
    \end{equation}
    there exists a unique solution $\eta^*$ to 
    \begin{equation}
         \E \left[ e^{ - \eta^* \,f(X_T)} \left( f(X_T)-c \right)\right]
         =0\,.
        \label{eq:find_eta_univariate}
    \end{equation}
\end{proposition}

\begin{proof}
Let $c$ be in the interval given in Equation \eqref{eq:c-interval} and assume there exists $a,b \in \R$, $a<b$, such that the cumulant generating function of $f(X_T)$, $K_{f(X_T)}(s)$, satisfies $K_{f(X_T)}(a), \, K_{f(X_T)}(b) < \infty$. 
The latter implies, by properties of the cumulant generating function, that $K_{f(X_T)}(s) < \infty$ for all $s \in (a,b)$. As $K_{f(X_T)}(s)$ is convex and twice differentiable, its derivative, $\frac{d}{ds}K_{f(X_T)}(s)$, exists and is continuously increasing on $(a,b)$. We observe that 
    \begin{equation*}
       \frac{d}{ds}K_{f(X_T)}(s) \Big |_{s = a} = \frac{\E\left[ e^{af(X_T)} f(X_T)\right]}{\E\left[ e^{af(X_T)}\right]}  \quad \text{and} \quad \frac{d}{ds}K_{f(X_T)}(s) \Big |_{s = b} = \frac{\E\left[ e^{bf(X_T)} f(X_T)\right]}{\E\left[ e^{bf(X_T)}\right]} \, .
    \end{equation*}
    Thus $\frac{d}{ds}K_{f(X_T)}(s)$ must take on all intermediate values between $\E\left[ e^{af(X_T)} f(X_T)\right]/\E\left[ e^{af(X_T)}\right]$ and \\$\E\left[ e^{bf(X_T)} f(X_T)\right]/\E\left[ e^{bf(X_T)}\right]$. Hence 
    there must exist an $-\eta^* \in (a,b)$ such that 
    \begin{equation}
    \label{eq:c-deriv}
            c = \frac{d}{ds}K_{f(X_T)}(s) \Big |_{s = - \eta^*} =  \frac{\E\left[ e^{ - \eta^* f(X_T)} f(X_T)\right]}{\E\left[ e^{-\eta^* f(X_T)}\right]} .
    \end{equation}
    Uniqueness of $\eta^*$ follows by increasingness of $\frac{d}{ds}K_{f(X_T)}(s)$. Finally, rearranging \eqref{eq:c-deriv} gives \eqref{eq:find_eta_univariate}.
\end{proof}

Note that if $f$ is constant, the condition in \cref{eq:c-interval} becomes void. In that case, one must have $c=f(X_T)$ for \cref{eq:find_eta_univariate} to have a solution.
An example when the assumptions in \cref{prop:c-range} are fulfilled is when $f(X_T)$ is bounded. In this case the cumulant generating function exists and is finite for all $s \in \R$. Thus, as a consequence of the above proposition we obtain that any stress $c$ satisfying $-\infty<\essinf f(X_T) < c < \esssup f(X_T)<+\infty$ yields a solution to \cref{eq:find_eta_univariate}. A case in point is VaR, which we discuss in detail in the \cref{sec:applications}. We further consider the case of two constraints, VaR and CVaR, and give explicit conditions under which a solution exists.

Under a stressed probability measure $\Q^*$, the stochastic process $X$ is no longer a true compound Poisson process due to the time- and state-dependence of the Girsanov kernel $h^*(t,x,y)$. Indeed, the next result states the intensity process and the severity distribution of $X$ under $\Q^*$.
\begin{proposition}
\label{prop:kappa_G}
The intensity process and severity distribution of the process $X$ under the stressed measure $\Q^*$ are, respectively,
\begin{equation*}
\kappa^*(t,x) =  \kappa \int_\mathbb{R} h^*(t,x,y) G(dy) 
\quad \text{and} \quad
    G^*(t,x,dy) = \frac{h^*(t,x,dy) G(dy)}{\int_\R  h^*(t,x,dy') G(dy')}.
\end{equation*}
\end{proposition}

\begin{proof}
By Girsanov's Theorem, the compensator of $\mu(dy,dt)$ under $\Q^*$ is $\nu^*(dy,dt)=h^*(t,x,y)\,\nu(dy,dt)$ and thus the intensity of the process becomes:
\begin{equation*}
\kappa^*(t,x) := \lim_{\Delta t \to 0}\frac{1}{\Delta t} \int_t^{t+\Delta t} \int_\mathbb{R} \nu^*(dy,ds) = \lim_{\Delta t \to 0}\frac{1}{\Delta t} \int_t^{t+\Delta t} \int_\mathbb{R} h^*(t,x,y) \kappa G(dy) dt =  \kappa \int_\mathbb{R} h^*(t,x,y) G(dy) . 
\end{equation*}
We obtain the updated severity distribution from $\nu^*(dy,dt)=h^*(t,x,y) \,\kappa G(dy) dt$, by multiplying and dividing by the integral term in $\kappa^*(t,x)$.
\end{proof}
Under the stressed measure $\Q^*$, we refer to the intensity $\kappa^*$ and the severity distribution $G^*$ as the stressed intensity and stressed severity distributions, respectively. Note that while in \cref{prop:kappa_G} we state the dynamics of $X$ under the optimal measure $\Q^*$, the result holds for any fixed and not necessarily optimal Lagrange multiplier $\boldeta$, $h^\boldeta$, and corresponding $\Q^\boldeta$. 

\cref{thm:optim_RN} gives a concise representation of the RN derivative of the stressed measure $\Q^*$ as a function of the value of the process at terminal time only, i.e., of $X_T$. This result, that the RN derivative of the stressed measure induced by perturbing a stochastic process is a function only of the terminal value of the process, is a priori not obvious. However, this representation of the RN derivative does not explicitly characterize the dynamics of the process under $\Q^*$. Indeed, the stressed intensity process and severity distribution, see \cref{prop:kappa_G}, are functions of the optimal random field $h^*$ given in \cref{thm:h}. We moreover provide in \cref{sec:numerics} an algorithm to simulate sample paths of the process under the stressed probability measure.

\section{Risk Measure Constraints}
\label{sec:applications}

In this section, we apply our framework to a risk management context. Specifically, we consider constraints given by the two most widely used risk measures, Value-at-Risk (VaR) and Conditional Value-at-Risk (CVaR). For a constraint on VaR and joint constraints on VaR and CVaR, we examine the dynamics of $X$ under the stressed measure both analytically and numerically.

\subsection{Value-at-Risk}
We first examine a stress on VaR. The VaR of a random variable $Z$ under a measure $\Q$ at level $\alpha \in [0,1]$ is defined as
\begin{equation*}
    \text{VaR}^\Q_\alpha(Z) = \inf \left\{ z \in \R \, | \,  F^\Q_Z(z) \geq \alpha \right\},
\end{equation*}
where $F^\Q_Z$ denotes the distribution function of $Z$ under $\Q$ and we use the convention that $\inf \emptyset = + \infty$. For notational simplicity, we write $\text{VaR}_\alpha(Z) := \text{VaR}^\P_\alpha(Z)$.

We implement a $\VaR_\alpha$ constraint in our framework by setting $f(x) = \mathds{1}_{\{x < q\}}$ and $c = \alpha \in (0,1)$ in \cref{opti:main}, which gives the constraint
\begin{equation}\label{eq:VaR-constraint}
    \Q(X_T < q) = \alpha.
\end{equation}
When the distribution of $X_T$ under $\P$ is continuous to the left of $q$, then \eqref{eq:VaR-constraint} is equivalent to requiring $\text{VaR}^\Q_\alpha(X_T) = q$; otherwise, it may be viewed as a probability constraint. Furthermore, $q$ must be chosen such that the constraint is attainable, meaning that $q$ must lie between the essential infimum and essential supremum of $X_T$, a requirement on the choice of $q$ rather than the process $X$.

\begin{proposition}
\label{prop:eta_h_VaR}
Let $\alpha \in (0,1)$ and $\essinf X_T < q < \esssup X_T$. The solution to \cref{opti:main} with constraint given by $f(x) = \mathds{1}_{\{x < q\}}$ and $c = \alpha$ is the measure $\Q^*$ characterized by the Girsanov kernel 
\begin{equation}
    h^*(t,x,y) = \frac{1 + \left( e^{-\eta^*} - 1 \right) \P\left(X_T - X_{t-}  < q - x - y \right)}{1 + \left( e^{-\eta^*} - 1 \right) \P\left(X_T - X_{t-}  < q - x  \right)},
    \label{eqn:h-opt-VaR}
\end{equation}
where the Lagrange multiplier $\eta^*$ is given by
\begin{equation*}
    \eta^* = \log \left(\frac{(1-\alpha)\P(X_T < q)}{\alpha \, \P(X_T \geq q)}\right) \, .
\end{equation*}
The solution is unique.
\end{proposition}

\begin{proof}
Letting $f(X_T) = \mathds{1}_{\{X_T < q\}}$ in \cref{thm:h}, we have for any $\eta$ that
\begin{align*}
h^\eta(t,x,y) &= \frac{\E_{t,x+y}\left[e^{-\eta\,\mathds{1}_{\{X_T < q\}}}\right]}{\E_{t,x}\left[e^{-\eta\,\mathds{1}_{\{X_T < q\}}}\right]} \\
&= \frac{\E_{t,x+y}\left[1 + \left( e^{-\eta} - 1 \right) \mathds{1}_{\{X_T < q\}} \right]}{\E_{t,x}\left[1 + \left( e^{-\eta} - 1 \right) \mathds{1}_{\{X_T < q\}} \right]} \\
&= \frac{1 + \left( e^{-\eta} - 1 \right) \P\left(X_T - X_{t-}  < q - x - y \right)}{1 + \left( e^{-\eta} - 1 \right) \P\left(X_T - X_{t-}  < q - x  \right)} .
\end{align*}
We find the optimal Lagrange multiplier $\eta^*$ by applying \cref{prop:optim_eta_gen}. For each $h^\eta$, the corresponding measure $\Q^\eta$ has RN derivative given in  \cref{thm:optim_RN}. Thus, the
constraint \eqref{eq:VaR-constraint} becomes
\begin{align*}
    0 =  \E \left[ e^{-\eta\mathds{1}_{\{X_T < q\}}}\left( \mathds{1}_{\{X_T < q\}}-\alpha \right)\right] 
    =  \E \left[ \left( e^{-\eta} \mathds{1}_{\{X_T < q\}} + \mathds{1}_{\{X_T \geq q\}} \right) \left( \mathds{1}_{\{X_T < q\}}-\alpha \right)\right].
\end{align*}
Solving for $\eta$ gives the result, which is well-defined by the choice of $\alpha$ and $q$.
\end{proof}

We note that since $\alpha = \Q(X_T < q)$, the optimal Lagrange multiplier $\eta^*$ is the log-odds ratio of the event $\{X_T < q\}$ under the reference and the stressed probability measures. If $q$ is chosen smaller than $\VaR_{\alpha}(X_T)$, corresponding to a downward stress on VaR, then it holds $\P(X_T < q) \leq \Q(X_T < q)$, and thus $\eta^* \leq 0$. Similarly, if $q$ is larger than $\VaR_{\alpha}(X_T)$, corresponding to an upward stress on VaR, we have $\P(X_T < q) \geq \Q(X_T < q)$, and thus $\eta^* \geq 0$.

The Girsanov kernel $h^*(t,x,y)$ provides information on how the process is distorted under the stressed measure.
More specifically, Proposition \ref{prop:kappa_G} states the stressed intensity and severity distributions, and that these are, in general, state and time dependent. The specific form of $h^*(t,x,y)$ in this example leads to a few observations summarized in the next two propositions.
\begin{proposition}
\label{prop:h_VaR}
Let $X$ be a compound Poisson process with intensity rate $\kappa$ and severity random variable $\xi$, where $\xi \stackrel{\P}{\sim} G$ has non-negative support.
Let $h^*$ be given in \cref{prop:eta_h_VaR}, i.e., the Girsanov kernel associated with the solution to \cref{opti:main} and constraints $f(x) = \mathds{1}_{\{x < q\}}$ and $c = \alpha$, where $\alpha \in (0,1)$ and $\essinf X_T < q < \esssup X_T$. Then, $h^*(t,x,y)$ satisfies for all $y \ge 0$,
\begin{equation*}
\lim_{t \to T}  h^* (t,x,y) =
\begin{cases}
1 \qquad &   \text{if}\quad x > q \, ,
\\
e^{\eta^*}+ \left( 1 - e^{\eta^*} \right) \mathds{1}_{\{x < q - y\}} \qquad & \text{if} \quad x \le q \, .
\end{cases}
\end{equation*}
Moreover, if $x > q$, then $h^*(t,x,y) = 1$ for all $t \in [0,T]$.
\end{proposition}
\begin{proof}
If $x > q$, then from the form of $h^*(t,x,y)$ given in \cref{prop:eta_h_VaR}, we observe that the probabilities in the numerator and denominator of $h^*(t,x,y)$ are both equal to 0 and thus $h^*(t,x,y)=1$ for any $t \in [0,T]$. 

Suppose that $x \leq q$. Using the properties of a Poisson process over small time intervals, we obtain an approximation of $\P\left(X_T - X_{t-}  \leq q - z  \right)$ for $z \in \R$ when $t$ is close to $T$. Using that under $\P$, $X_t = \sum_{i=1}^{N_t} \xi_i$, where $\xi_i$ are i.i.d. severity random variables with distribution $G$ and $N_t$ is a Poisson process with intensity $\kappa$,
we obtain for $\epsilon>0$ small
\begin{align*}
    \P\left(X_T - X_{T-\epsilon}  \leq q - z  \right) &= \P\left(\sum_{N_{T-\epsilon}}^{N_T} \xi_i  \leq q - z  \right) \\
    &= \P\left(N_{\epsilon} = 0 \right) + \P\left( \{ N_{\epsilon} = 1 \} \cap \{ \xi  \leq q - z  \} \right) + \P\left( \{ N_{\epsilon} > 1 \} \cap \{ \xi \leq q - z  \} \right) \\
    &= 1 - \kappa \epsilon + o(\epsilon) + (\kappa \epsilon + o(\epsilon)) \P\left(\xi  \leq q - z  \right) + o(\epsilon) \\
    &= 1 + \kappa \epsilon \left(G(q-z)  - 1 \right) + o(\epsilon). \numberthis \label{eqn:approx}
\end{align*}

We split the case $x \le q$ into two parts: first, we further assume that $x > q - y$, then the probability in the numerator of $h^*(t,x,y)$ is zero, thus substituting the approximation given in \eqref{eqn:approx} gives 
\begin{equation*}
    h^*(T- \epsilon,x,y) = \frac{1}{1 + \left( e^{-\eta^*} - 1 \right) \left[ 1 + \kappa \epsilon \left( G(q-x)  - 1 \right)+ o(\epsilon) \right]} .
\end{equation*}
When $\epsilon \to 0$, we have $h^*(T- \epsilon,x,y) \to e^{\eta^*}$.

For the second part we further assume that $x\leq q - y$. Substituting the short-time approximation \eqref{eqn:approx} into both the numerator and denominator of $h^*(t,x,y)$, gives
\begin{equation*}
    h^*(T- \epsilon,x,y) = \frac{1 + \left( e^{-\eta^*} - 1 \right) \left[ 1 + \kappa \epsilon \left( G(q-x-y)  - 1 \right) + o(\epsilon) \right]}{1 + \left( e^{-\eta^*} - 1 \right) \left[ 1 + \kappa \epsilon \left( G(q-x)  - 1 \right) + o(\epsilon) \right]},
\end{equation*}
which converges to 1 as $\epsilon \to 1$.
\end{proof}

\begin{proposition}
\label{prop:kappa_VaR}
Let $X$ be a compound Poisson process with rate $\kappa$ and severity random variable $\xi$, where $\xi \stackrel{\P}{\sim} G$ has non-negative support. Let $\kappa^*$ be the intensity process of $X$ under the stressed measure $\Q^*$ that solves \cref{opti:main} with constraints $f(x) = \mathds{1}_{\{x < q\}}$ and $c = \alpha$, where $\alpha \in (0,1)$ and $\essinf X_T < q < \esssup X_T$. Then $\kappa^*(t,x)$ satisfies
\begin{equation*}
\lim_{t \to T}  \kappa^*(t,x) =
\begin{cases}
\kappa \qquad &   \text{if}\quad x > q \, ,
\\
\kappa \left[ e^{\eta^*} + \left( 1 - e^{\eta^*} \right) G(q-x ) \right] \qquad & \text{if} \quad x \le q \, .
\end{cases}
\end{equation*}
Moreover, if $x > q$, then $\kappa^*(t,x) = \kappa$ for all $t \in [0,T]$, i.e., the stressed intensity is equal to the intensity under the reference measure $\P$.
\end{proposition}
\begin{proof}
If $x> q$, then by \cref{prop:h_VaR}, $h^*(t,x,y) = 1$ for all $t \in [0,T]$ and therefore $\kappa^*(t,x)= \int_{\R_+} \kappa \, G(dy) = \kappa$ for $t \in [0,T]$. Else, if $x \leq q$, we have by dominated convergence:
\begin{align*}
\lim_{t \to T}\kappa^*(t,x) 
&=
\kappa \int_{\R_+} \lim_{t \to T} h^*(t,x,y)  G(dy) \\
&= \kappa \int_{\R_+} \left[ e^{\eta^*} + \left( 1 - e^{\eta^*} \right) \mathds{1}_{\{x \leq q - y\}} \right]  G(dy) \\
&= \kappa \left[ e^{\eta^*}  + (1 - e^{\eta^*} )G(q-x) \right] .
\end{align*}
\end{proof}

Finally, in the following example we examine a numerical example of a 15\% upward stress on $\text{VaR}_\alpha(X_T)$.

\begin{figure}[htb!]
    \centering
    \includegraphics[width=10cm]{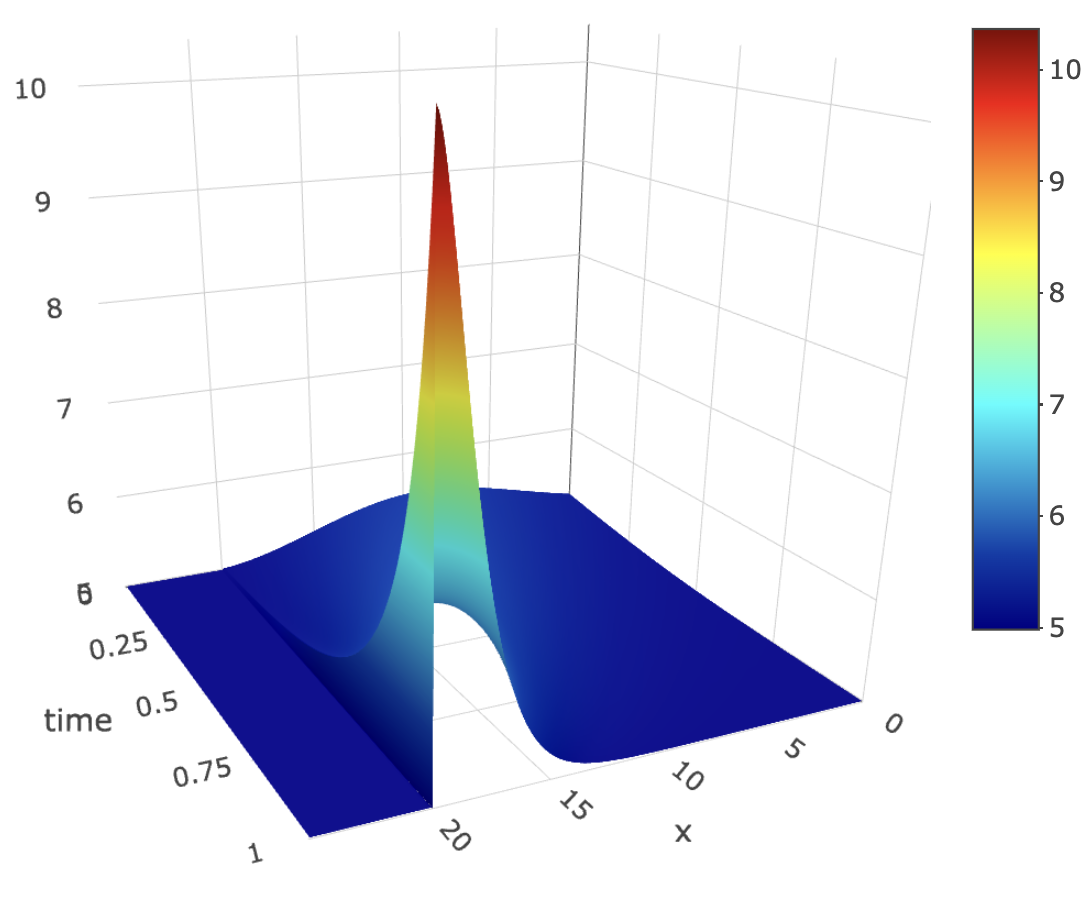}
    \caption{Stressed intensity $\kappa^*(t,x)$ for a 15\% stress upwards on $\text{VaR}_{0.9}(X_T)$. Parameters under $\P$ are $\kappa=5$ and $\xi \sim \Gamma(2,1)$, where $\E[\xi]=2$.}
    \label{fig:VaR_kappa_3d}
\end{figure}

\begin{example}
\label{ex:VaR}
Here we illustrate how the dynamics of a compound Poisson process are altered for a constraint on VaR. Specifically, under the reference measure $\P$, we consider a compound Poisson process $X$ over the time horizon $[0,T]$ where $T = 1$, intensity $\kappa = 5$, and the severity is Gamma distributed $\xi \sim \Gamma(a, b)$, where $a=2$ is the shape parameter and $b = 1$ is the rate parameter (i.e., $\E[\xi]=2$, $\text{Var}(\xi) = 2$). With these parameters, the 90\% quantile of $X_T$ is $\text{VaR}_{0.9}(X_T) = 17.4$. We apply a stress composed of a 15\% increase in VaR at the 90\% level, resulting in $\text{VaR}^\Q_{0.9}(X_T)$ of approximately 20. Simulation under the reference and the stressed measure is performed using methods described in \cref{sec:numerics}, which requires $h^*$ in \eqref{eqn:h-opt-VaR} to be plugged into Proposition \ref{prop:kappa_G}.

\cref{fig:VaR_kappa_3d} shows the optimal intensity $\kappa^*(t,x)$ for times $t \in [0,1]$ and $x\in [0,25]$. We see that if $x$ is greater than 20, then $\kappa^*(t,x) = \kappa = 5$, i.e., the stressed intensity is the same as under the reference measure. We further observe at time $T= 1$, when $x$ is close to 0, $\kappa^*(t,x)$ is close to $\kappa = 5$, and as $x$ approaches $q$, $\kappa^*(T,x)$ approaches $\kappa e^{\eta^*}=10.37$. These observations are consistent with \cref{prop:kappa_VaR}. 

\begin{figure}[bth!]
    \centering
    \includegraphics[width=6in]{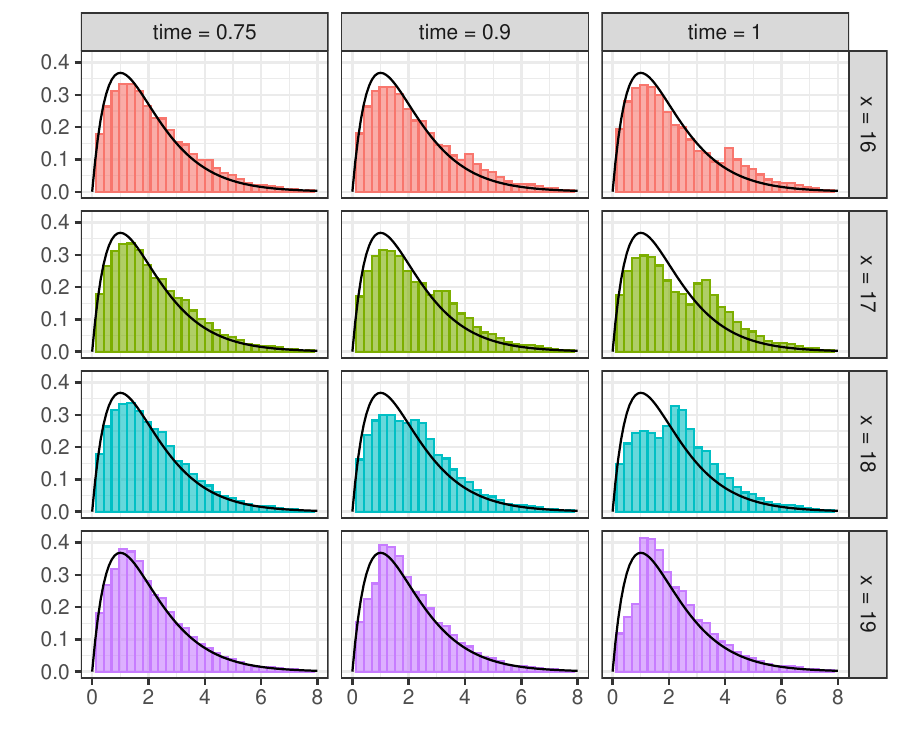}
    \caption{Histogram of 1 million draws from the stressed severity distribution $G^*(t,x,dy)$ under a 15\% increase in $\text{VaR}_{0.9}(X_T)$ for times $t =$ 0.75 (left column), $t=0.9$ (middle column), and $t=1$ (right column) and state space $x =$ 16, 17, 18, and 19 (top to bottom). 
    The black lines are the density of the severity distribution under $\P$, which is $\Gamma(2,1)$ for all values of time and state space.}
    \label{fig:VaR_G}
\end{figure}

\cref{fig:VaR_G} shows histograms of the stressed severity distribution $G^*(t,x,dy)$ for state space values $x \in \{16,17,18,19 \}$ and times $t \in \{ 0.75, 0.9, 1\}$. The black line depicts the density of the severity random variable under $\P$, which is $\Gamma(2,1)$. Consistent with \cref{prop:h_VaR}, we see that the severity distributions are more distorted when $x$ is close to, but less than, $q$. The change in the distribution also become more apparent as we approach terminal time. The distribution remains fairly close to the baseline at $t=0.75$ (left column), while we observe a significant change in the distribution at time 0.9 (middle column) and in particular at the terminal time, 1 (right column). At the terminal time, when $x=$ 18 and 19, we observe that the stressed severity distributions has a larger right tail which is induced by the stress, while for $x=16$ and 17, the distribution is bimodal.

Select sample paths of $X$ under the stressed measure and their corresponding stressed intensity processes are shown in \cref{fig:VaR_paths}. The black lines are the 10th, 50th, and 90th quantiles under both $\P$ (solid line) and $\Q^*$ (dashed line). The dashed grey line is $q = \text{VaR}^\Q_{0.9}(X_T) = 19.97$. Examine first the blue line -- once the path crosses $q$ (the dashed grey line) in the left panel, the corresponding stressed intensity process (right panel) falls to the value of the intensity under the reference measure, i.e., $\kappa = 5$. In contrast, the yellow path (left panel) never crosses $q$, and thus $\kappa^*(t,x)$ continues to increase as $t \to T$. Paths which never get close to the value of $q$ have intensity processes that are not extensively distorted, i.e., $\kappa^*(t,x)$ is close to $\kappa=5$. This is consistent with our observations of \cref{fig:VaR_kappa_3d}.

\begin{figure}[htb!]
\centering
\subfloat[Paths $X_t, \, t \in {[0,1]}$ under $\Q^*$]{\includegraphics{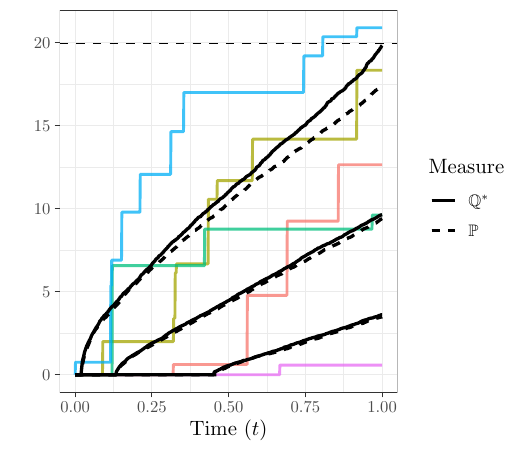}}
\subfloat[Stressed intensity process $\kappa^*(t,x)$]{\includegraphics{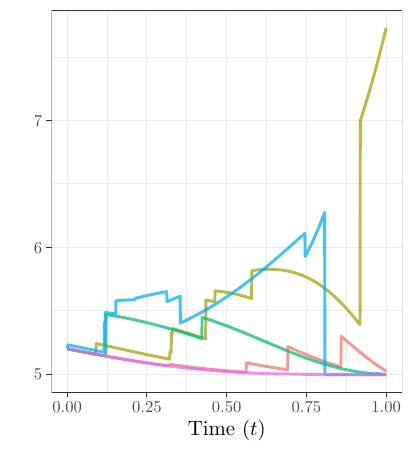}}
\caption{Panel (a) displays sample paths under the stressed measure. 10,000 paths were simulated and five are highlighted along with the 10\%, 50\%, and 90\% quantiles (black lines) for $\P$ and $\Q^*$. Panel (b) displays the corresponding (same colour) sample path of the  stressed intensity processes. Under $\P$, the severity is $\xi \sim \Gamma(2,1)$ and the intensity is $\kappa=5$. A stress of a 15\% increase in $\text{VaR}_{0.9}(X_T)$ is applied.}
\label{fig:VaR_paths}
\end{figure}
\end{example}

\subsection{Conditional Value-at-Risk}
Next, we examine the effect on the dynamics of $X$ under a stress on the risk measure CVaR. The CVaR, also called Expected Shortfall, of a random variable $Z$ under a measure $\Q$ at level $\alpha \in [0,1)$ is \citep{acerbi2002coherence}
\begin{equation*}
    \text{CVaR}^\Q_\alpha(Z) = \frac{1}{1-\alpha} \int_\alpha^1 \text{VaR}^\Q_u(Z) \, du\,.
\end{equation*}
For notational simplicity, we write $\text{CVaR}_\alpha(Z) := \text{CVaR}^\P_\alpha(Z)$.  

Since the value of CVaR$_\alpha^\Q$ depends on that of VaR$_\alpha^\Q$, we impose constraints on both risk measures. We exploit the representation 
\begin{equation*}
    \text{CVaR}^\Q_\alpha(Z) = \frac{1}{1-\alpha} \E[(X-\text{VaR}^\Q_\alpha(Z))_+] + \text{VaR}^\Q_\alpha(Z)
\end{equation*}
to write this optimization problem in the form of \cref{opti:main} using the following two constraints:
\begin{equation}
\label{eq:CVaR-constraint}
    \Q(X_T < q) = \alpha \quad \text{and} \quad \E^\Q[(X_T-q)_+] = (s-q)(1-\alpha).
\end{equation}
The solution to \cref{opti:main} with the above constraints is given in the next proposition.

\begin{proposition}
\label{prop:h_eta_CVaR} Let $K_{X_T | X_T > q}(s)$ denote the cumulant generating function of $X_T|X_T > q$. Suppose $\alpha \in (0,1)$, $\essinf X_T < q < s < \esssup X_T$, and there exists an $r>0$ such that $K_{X_T | X_T > q}(s)<\infty$ for all $s \in (-r,r)$. Then the solution to \cref{opti:main} with constraints $f_1(x) = \mathds{1}_{\{x < q\}}$, $c_1 = \alpha$, $f_2(x) = (x-q)\mathds{1}_{\{x \geq q\}}$, and $c_2 = (s-q)(1-\alpha)$ is the measure $\Q^*$ characterized by the Girsanov kernel
\begin{equation*}
    h^*(t,x,y) = \frac{ e^{-\eta_1^*} \, \P\left(X_T - X_{t-}  \leq q - x - y \right) + e^{-\eta_2^*(x+y-q)} \E\left[  e^{-\eta_2^*(X_T - X_{t-})}  \mathds{1}_{\{X_T -X_{t-} \geq q -x-y\}}\right] }{ e^{-\eta_1^*} \, \P\left(X_T - X_{t-}  \leq q - x  \right) + e^{-\eta_2^*(x-q)} \E\left[  e^{-\eta_2^*(X_T - X_{t-})}  \mathds{1}_{\{X_T -X_{t-} \geq q -x\}}\right] },
\end{equation*}
where $\eta_2^*$ is the solution to
\begin{equation}
\label{eq:eta2_eq}
    \E \left[ e^{-\eta_2^* (X_T-q)} (X_T-s) \mathds{1}_{\{X_T > q\}} \right] = 0
\end{equation}
and
\begin{equation*}
    \eta_1^* = \log \left(\frac{(1-\alpha)}{\alpha}\frac{\P(X_T < q)}{\E \left[ e^{-\eta_2 (X_T-q)} \mathds{1}_{\{X_T \geq q\}} \right]}\right) = \eta^\text{VaR} + \log \left(\frac{\P(X_T \geq q)}{\E \left[ e^{-\eta_2 (X_T-q)} \mathds{1}_{\{X_T \geq q\}} \right]}\right).
\end{equation*}
Here, $\eta^\text{VaR}$ refers to the optimal Lagrange multiplier when only a VaR constraint is applied (\cref{prop:eta_h_VaR}). The optimization problem has a unique solution.
\end{proposition}

\begin{proof}
First, we find the Girsanov kernel for given Lagrange multipliers $\boldeta=(\eta_1,\eta_2)$. We compute
\begin{align*}
    \E_{t,x}\left[\exp \left( -\eta_1\,f_1(X_T) - \eta_2 f_2(X_T)\right)\right] &= \E_{t,x}\left[e^{-\eta_1 \Id_{\{X_T < q\}}  - \eta_2 (X_T-q) \Id_{\{X_T \geq q\}}} \right] \\
    &= \E_{t,x}\left[e^{-\eta_1} \Id_{\{X_T < q\}}  + e^{- \eta_2 (X_T-q)} \Id_{\{X_T \geq q\}} \right] \\
     &= e^{-\eta_1} \P (X_T < q | X_{t-}=x) + e^{\eta_2 q} \, \E \left[ e^{- \eta_2 X_T} \Id_{\{X_T \geq q\}} | X_{t-}=x \right] \\
    &= e^{-\eta_1} \P (X_T - X_{t-} < q-x) + e^{-\eta_2 (x-q)}  \E \left[ e^{- \eta_2 (X_T-X_{t-})} \Id_{\{X_T-X_{t-} \geq q - x \}}  \right] \, ,
\end{align*}
and the formula for $h^*$, if the optimal Lagrange multipliers exist, follows from \cref{thm:h}. Next, we show the existence of the Lagrange multipliers and their representation.
By \cref{thm:optim_RN}, the RN derivative of the measure $\Q^\boldeta$ that solves \cref{opti:main} with constraints \eqref{eq:CVaR-constraint} is
\begin{equation*}
    \frac{d\Q^\boldeta}{d\P} = \frac{e^{-\eta_1 f_1(X_T)-\eta_2 f_2(X_T)}}{\E[e^{-\eta_1 f_1(X_T)-\eta_2 f_2(X_T}]} =\frac{e^{-\eta_1} \mathds{1}_{\{X_T < q\}} + e^{-\eta_2 (X_T - q)} \mathds{1}_{\{X_T \geq q\}}}{\E[e^{-\eta_1} \mathds{1}_{\{X_T < q\}} + e^{-\eta_2 (X_T - q)} \mathds{1}_{\{X_T \geq q\} }]}  .
\end{equation*}
Using this representation of the RN derivative, the first constraint can be written as
\begin{align*}
    0 = \E \left[ \left( e^{-\eta_1} \mathds{1}_{\{X_T < q\}} + e^{-\eta_2 (X_T - q)} \mathds{1}_{\{X_T \ge q\}} \right) \left(\mathds{1}_{\{X_T < q\}} - \alpha \right) \right].
\end{align*}
Solving for $\eta_1$ gives 
\begin{equation}
\label{eq:eta1}
    \eta_1 = \log\left( \frac{1-\alpha}{\alpha} \frac{\P(X_T < q)}{\E[e^{-\eta_2 (X_T - q)} \mathds{1}_{\{X_T \geq q\}}]} \right).
\end{equation}
Analogously, we may write the second constraint as
\begin{align*}
    0 = \E \left[ \left( e^{-\eta_1} \mathds{1}_{\{X_T < q\}} + e^{-\eta_2 (X_T - q)} \mathds{1}_{\{X_T \geq q\}} \right) \Big((X_T - q)\mathds{1}_{\{X_T \geq q\}} - (s-q)(1-\alpha) \Big) \right].
\end{align*}
Simplifying this expression and substituting in the above expression for $\eta_1$ gives 
\begin{equation}\label{eq:eta-2-eqn}
   0 = \E \left[ e^{-\eta_2 (X_T-q)} (X_T-s) \mathds{1}_{\{X_T \geq q\}} \right], 
\end{equation}
providing \eqref{eq:eta2_eq}.
To determine if \eqref{eq:eta-2-eqn} has a solution, we first rewrite for fixed $\boldeta$ the RN derivative using \eqref{eq:eta1}:
\begin{equation*}
    \frac{d\Q^\boldeta}{d\P} =  \frac{\alpha}{\P(X_T < q)} \mathds{1}_{\{X_T < q\}} + \frac{1-\alpha}{\E[e^{-\eta_2 (X_T - q)} \mathds{1}_{\{X_T \geq q\}}]} e^{-\eta_2 (X_T - q)} \mathds{1}_{\{X_T \geq q\}} .
\end{equation*}
Therefore, the second constraint becomes
\begin{align*}
    (s-q)(1-\alpha) = \E^{\Q^\boldeta} \left[ (X_T-q)_+ \right] = \E \left[  \frac{1-\alpha}{\E[e^{-\eta_2 (X_T - q)} \mathds{1}_{\{X_T \geq q\}}]} e^{-\eta_2 (X_T - q)} (X_T - q) \mathds{1}_{\{X_T \geq q\}}  \right] .
\end{align*}
Simplifying the above expression gives 
\begin{equation}
\label{eq:eta2_eq_key}
    s-q 
    = \frac{\E \left[ e^{-\eta_2 (X_T - q)} (X_T - q)_+ \right]}{\E [e^{-\eta_2 (X_T - q)}\mathds{1}_{\{X_T > q\}} ]} 
    = \frac{\E \left[ e^{-\eta_2 (X_T - q)} (X_T - q) | X_T > q \right]}{\E [e^{-\eta_2 (X_T - q)}  | X_T > q]} .
\end{equation}

We observe that the right-hand side of \eqref{eq:eta2_eq_key} is the derivative of the cumulant generating function of $X_T - q | X_T > q$ evaluated at $-\eta_2$. We will denote the derivative of the cumulant generating function of $X_T - q | X_T > q$ by $\frac{d}{ds}K_{X_T - q | X_T > q}(s)$.  Since moment generating functions are log-convex, $\frac{d}{ds}K_{X_T - q | X_T > q}(s) \big |_{s=-\eta_2}$ is decreasing in $\eta_2$.
To show that there exists a solution to \cref{{eq:eta2_eq_key}}, we denote $\alpha^\P= \P(X_t < q)$ and $s^\P = \E [X_T | X_T > q]$.

Consider the following cases: first, if $s = s^\P$, then $\eta_2 = 0$ solves \eqref{eq:eta2_eq_key}. Second, assume that $s > s^\P$. Observe that when $\eta_2 = 0$, the right-hand side of \eqref{eq:eta2_eq_key} is $\E [X_T - q | X_T > q] = s^\P-q < s-q$ by assumption. Moreover, as $\eta_2 \to - \infty$, $\frac{d}{ds} K_{X_T - q | X_T > q}(s) \big |_{s=-\eta_2}$ is monotonically increasing, and in fact,
\begin{equation*}
    \lim_{\eta_2 \to - \infty} \frac{d}{ds}K_{X_T - q | X_T > q}(s) \big |_{s=-\eta_2} = \esssup X_T - q > s-q .
\end{equation*}
Thus we conclude \eqref{eq:eta2_eq_key} must have a solution for some $\eta_2 < 0$. 

Finally, suppose $s < s^\P$. When $\eta_2 = 0$, the right-hand side of \eqref{eq:eta2_eq_key} is $\E [X_T - q | X_T > q] > s - q$.  Since as $\eta_2 \to  \infty$, $\frac{d}{ds} K_{X_T - q | X_T > q}(s) \big |_{s=-\eta_2}$ is monotonically decreasing and
\begin{equation*}
    \lim_{\eta_2 \to \infty} \frac{d}{ds} K_{X_T - q | X_T > q}(s) \big |_{s=-\eta_2} = \essinf X_T - q< s-q,
\end{equation*}
we conclude \eqref{eq:eta2_eq_key} must have a solution for some $\eta_2 > 0$.
\end{proof}

We note that the sign of $\eta_2^*$ is as follows:
\begin{corollary}
For $\alpha \in (0,1)$, let $q$ and $s$ be the stressed $\VaR_\alpha$ and $\text{CVaR}_\alpha$ values, respectively. Define $\alpha^\P= \P(X_T < q)$ and $s^\P = \E [X_T | X_T > q]$. If $s < s^\P$, then $\eta_2^* > 0$, and if $s > s^\P$, then $\eta_2^* < 0$.
\end{corollary}
This follows directly from the proof of \cref{prop:h_eta_CVaR}. If $X_T$ is continuous,  $s^\P = \text{CVaR}_{\alpha^\P}(X_T)$.
Thus, the sign of $\eta_2^*$ depends in a direct manner on the stress on CVaR while its dependence on the stress on VaR is indirectly through $\alpha^\P$.

\begin{proposition}
Let the severity distribution $G$ have non-negative support. Let $h^*$ be given in \cref{prop:h_eta_CVaR}, i.e., the Girsanov kernel associated with the solution to \cref{opti:main} with constraints \eqref{eq:CVaR-constraint}, where $\alpha \in (0,1)$, $\essinf X_T < q < s < \esssup X_T$, and there exists an $r>0$ such that $K_{X_T | X_T > q}(s)<\infty$ for all $s \in (-r,r)$. Then for all $y \ge 0$, $h^*(t,x,y)$ satisfies
\begin{equation*}
\lim_{t \to T}  h^* (t,x,y) =
\begin{cases}
e^{-\eta_2^* y} \qquad &   \text{if}\quad x > q \, ,
\\
\mathds{1}_{\{y < q - x\}} + e^{\eta_1^* -\eta_2^* (x+y-q)} \mathds{1}_{\{y \geq q-x\}}  \qquad  & \text{if} \quad x \le q \, .
\end{cases}
\end{equation*}
Moreover, if $x > q$, then $h^*(t,x,y) = e^{-\eta_2^* y}$ for any $t \in [0,T]$.
\end{proposition}

\begin{proof}
First, if $x > q$, then $q-x < 0$ and $q-x-y < 0$, thus the probabilities in the numerator and denominator of $h^*(t,x,y)$ from \cref{prop:h_eta_CVaR} are both 0, and the indicator functions in the expectations are both identically equal to 1. This gives
\begin{equation*}
    h^*(t,x,y) = \frac{e^{-\eta_2^*(x+y-q)} \E \left[  e^{-\eta_2^*(X_T - X_{t-})}  \right] }{ e^{-\eta_2^*(x-q)} \E \left[  e^{-\eta_2^*(X_T - X_{t-})} \right] }  = e^{-\eta_2^* y} .
\end{equation*}
Second, if $x \leq q$, we examine the limiting behaviour of $h^*(t,x,y)$. In addition to the approximation of $\P\left(X_T - X_{t-}  \leq q - z  \right)$ derived in the proof of \cref{prop:h_VaR}, we further use the properties of a Poisson process over short intervals to obtain an approximation of $\E\left[e^{-\eta_2(X_T - X_{t-})}  \mathds{1}_{\{X_T -X_{t-} \geq q - z\}}\right] $ for $t = T-\epsilon$, $\epsilon>0$, as follows:
\begin{align*}
    \E\left[  e^{-\eta_2^*(X_T - X_{T-\epsilon})}  \mathds{1}_{\{X_T -X_{T-\epsilon} \geq q - z\}}\right]  &= \E\left[  e^{-\eta_2^* \xi}  \mathds{1}_{\{\xi \geq q - z\}}\right] \P(N_\epsilon = 1) + \E\left[\mathds{1}_{\{0 \geq q - z\}}\right] \P(N_\epsilon = 0 ) + o(\epsilon) \\
    &= \E\left[  e^{-\eta_2^* \xi}  \mathds{1}_{\{\xi \geq q - z\}}\right] (\kappa \epsilon + o(\epsilon)) + \E\left[\mathds{1}_{\{0 \geq q - z\}}\right] (1 -\kappa \epsilon + o(\epsilon)) + o(\epsilon) \\
    &= \kappa \epsilon \left( \E\left[  e^{-\eta_2^* \xi}  \mathds{1}_{\{\xi \geq q - z\}}\right] - \mathds{1}_{\{z \geq q \}} \right) + \mathds{1}_{\{z \geq q \}} + o(\epsilon) \, .
\end{align*}
Using these two approximations, we consider two cases: first, if $q-y < x \leq q$, then for $t = T-\epsilon$ we have
\begin{align*}
    h^*(T- \epsilon,x,y) = \frac{  e^{-\eta_2^*(x+y-q)} \left[\kappa \epsilon \left( \E\left[  e^{-\eta_2^* \xi}  \mathds{1}_{\{\xi \geq 0\}}\right] - 1 \right) + 1 \right] + o(\epsilon)}{ e^{-\eta_1^*} \, \left(1 + \kappa \epsilon \left( G(q-x)  - 1 \right)  \right) + e^{-\eta_2^*(x-q)} \left( \kappa \epsilon \, \E\left[  e^{-\eta_2^* \xi}  \mathds{1}_{\{\xi \geq q - x\}}\right] \right) + o(\epsilon) } \,.
\end{align*}
Taking the limit, we obtain
\begin{equation*}
    \lim_{\epsilon \to 0} h^*(T- \epsilon,x,y) = \frac{  e^{-\eta_2^*(x+y-q)} }{ e^{-\eta_1^*} } = e^{\eta_1^* -\eta_2^*(x+y-q)}\,.
\end{equation*}
Else, if $x\leq q - y$, we have
\begin{align*}
    h^*(T- \epsilon,x,y) 
    =
    \frac{ e^{-\eta_1^*} \, \left(1 + \kappa \epsilon \left( G(q-x-y)  - 1 \right)  \right) + e^{-\eta_2^*(x+y-q)} \, \kappa \epsilon \, \E\left[  e^{-\eta_2^* \xi}  \mathds{1}_{\{\xi \geq q - x -y \}}\right]  + o(\epsilon)}
    { e^{-\eta_1^*} \, \left(1 + \kappa \epsilon \left( G(q-x)  - 1 \right)  \right) + e^{-\eta_2^*(x-q)} \,\kappa \epsilon \, \E\left[  e^{-\eta_2^* \xi}  \mathds{1}_{\{\xi \geq q - x\}}\right]  + o(\epsilon) } \,.
\end{align*}
Taking the limit, we obtain
\begin{equation*}
    \lim_{\epsilon \to 0} h^*(T- \epsilon,x,y) = \frac{  e^{-\eta_1^*}}{ e^{-\eta_1^*} } = 1 .
\end{equation*}
\end{proof}

\begin{proposition}
\label{prop:kappa-CVaR}
Let the severity random variable $\xi$ with $\P$-distribution $G$ have non-negative support. Then, the limiting behaviour of $\kappa^*(t,x)$ is
\begin{equation*}
\lim_{t \to T}  \kappa^*(t,x) =
\begin{cases}
\kappa \,\E[e^{-\eta_2^* \xi}] \qquad  &   \text{if}\quad x > q \, ,
\\[0.5em]
\kappa \left( G(q-x) + e^{\eta_1^* - \eta_2^* (x-q)} \E[e^{-\eta_2^* \xi} \mathds{1}_{\{\xi \geq q-x\}}] \right) \qquad & \text{if} \quad x \le q \, .
\end{cases}
\end{equation*}
Moreover, if $x > q$, the stressed intensity is $\kappa^*(t,x) = \kappa\, \E[e^{-\eta_2^* \xi}]$ for any $t \in [0,T]$. 
\end{proposition}

\begin{proof}
If $x> q$, we have $h^*(t,x,y) = e^{-\eta_2^* y}$ by the previous proposition, so
\begin{equation*}
\kappa^*(t,x) = \kappa \int_{\R_+} e^{-\eta_2^* y}  G(dy) = \kappa\; \E[e^{-\eta_2^* \xi}].
\end{equation*}
If $x \leq q$, we apply dominated convergence to obtain
\begin{align*}
\lim_{t \to T}\kappa^*(t,x) 
&=
\int_{\R_+} \lim_{t \to T} h^*(t,x,y) \kappa \, G(dy) \\
&=\kappa \int_{\R_+} \left[ \mathds{1}_{\{y \leq q - x\}} + e^{\eta_1^* -\eta_2^*(x+y-q)} \mathds{1}_{\{y > q - x\}} \right]  G(dy) \\
&= \kappa \left( G(q-x) + e^{\eta_1^* -\eta_2^*(x-q)} \E\left[e^{-\eta_2^* \xi} \mathds{1}_{\{\xi > q-x\}} \right] \right).
\end{align*}

\end{proof}

If the severity distribution under the reference measure $\P$ is explicitly given, we may obtain the limit of $\kappa^*(t,x)$ for $t \to T$ analytically. In keeping with our recurring example, let $\xi \sim \Gamma(a, b)$ under $\P$, where $a$ is the shape and $b$ the rate parameter. Then the limiting behaviour of $\kappa^*(t,x)$, if $\eta_2^* > -b$, is
\begin{equation*}
\lim_{t \to T}  \kappa^*(t,x) =
\begin{cases}
\kappa \left(\frac{b}{b + \eta^*_2} \right)^a \qquad &   \text{if}\quad x > q \, ,
\\
\kappa \left( \mathbb{P}(\xi < q-x) + e^{\eta_1^*-\eta_2^* (x - q)} \left(\frac{b}{b + \eta_2^*} \right)^a  \mathbb{P}(\xi' \geq q-x) \right) \qquad  & \text{if} \quad x \le q \,,
\end{cases}
\end{equation*}
where $\xi ' \sim \Gamma(a, b + \eta_2^*)$.
\begin{figure}[t!]
    \centering
    \includegraphics[width=10cm]{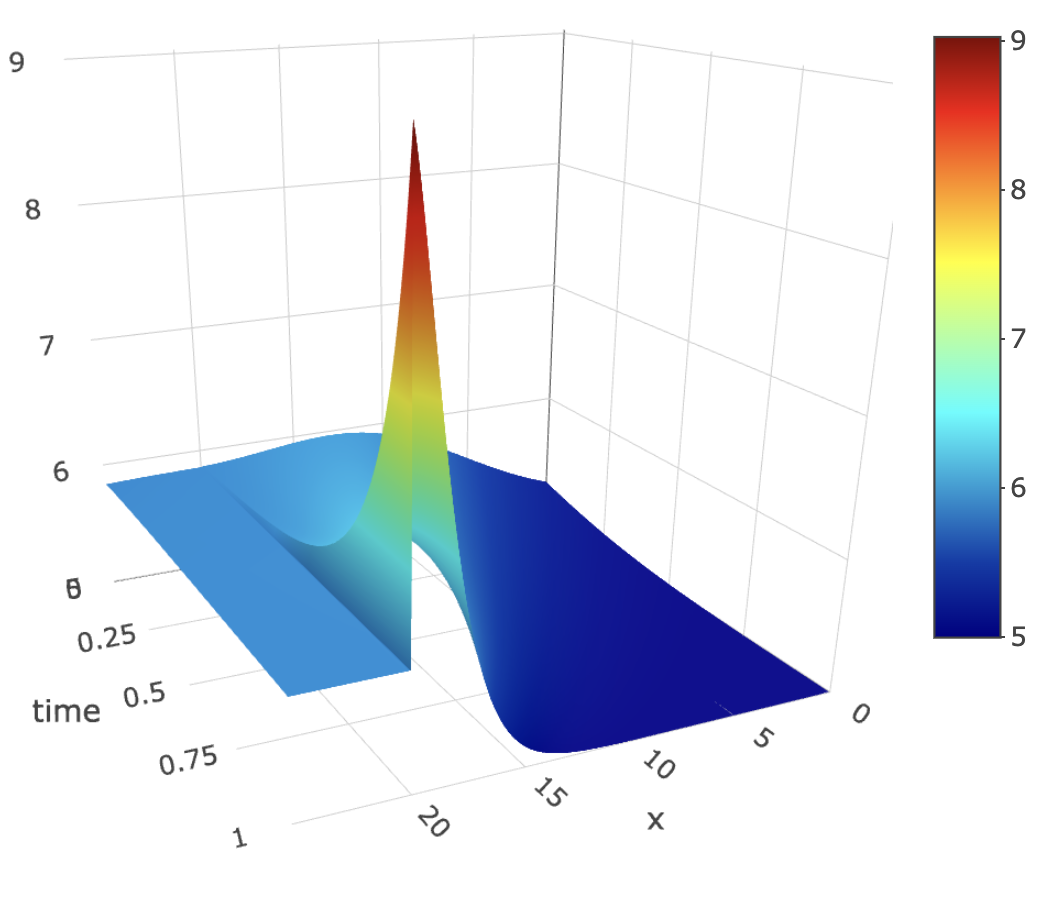}
    \caption{Stressed intensity $\kappa^*(t,x)$ for constraints corresponding to a 15\% increase in $\text{VaR}_{0.9}(X_T)$ and a 12\% increase in $\text{CVaR}_{0.9}(X_T)$. Parameters under $\P$ are $\kappa=5$ and $\xi \sim \Gamma(2,1)$, where $\E[\xi]=2$}
    \label{fig:CVaR_kappa_3d}
\end{figure}

Finally, we examine the effect of joint stresses on both VaR and CVaR in the following numerical example:

\begin{example}
\label{ex:CVaR} 
Let $X$ be under $\P$ a compound Poisson process over a time horizon $[0,1]$ with severity distribution $\xi \sim \Gamma(2, 1)$ and intensity $\kappa = 5$, as in \cref{ex:VaR}. We apply a stress of a 15\% increase in VaR at the 90\% level, as in \cref{ex:VaR}, and a 12\% increase in the CVaR at the 90\% level. 

\cref{fig:CVaR_kappa_3d} displays $\kappa^*(t,x)$ for a grid of times $t \in [0,1]$ and state space values $x$ in [0,25]. We see that if $x$ is greater than $q = \text{VaR}^\Q_{0.9}(X_T) = 19.97$, then consistent with \cref{prop:kappa-CVaR}, $\kappa^*(t,x) = \kappa \left(b / (b + \eta^*_2)\right)^a = 5.87$. If $x$ is less than $q$, we observe that at time $T= 1$, when $x$ is close to 0, $\kappa^*(t,x)$ is close to $\kappa = 5$, and as $x$ approaches $q$ from below, $\kappa^*(t,x)$ approaches $\kappa e^{\eta^*_1} \left(\frac{b}{b + \eta^*_2} \right)^a = 9.08$.

\cref{fig:CVaR_G} displays histograms of the stressed severity distribution $G^*(t,x,dy)$  for state space values $x \in \{ 16,17, 18,19 \}$ and times $t \in \{ 0.75, 0.9, 1\}$. The black lines are the density of the severity distribution under $\P$, here $\Gamma(2,1)$. As in the VaR stress case (\cref{ex:VaR}), we observe that the stressed severity distributions become more and more distorted from the reference distribution as time approaches $t=1$. We note in this case, however, that the addition of the CVaR constraint has a smoothing effect on the tails. Compared to \cref{fig:VaR_G}, there is a reduced  bimodal effect for $x=$ 16 and 17, and less of a spike around 1, when $x=$ 18 and 19.

\cref{fig:CVaR_paths} shows select sample paths of $X$ under $\Q^*$ (left panel) and their corresponding intensity processes $\kappa^*(t,x)$ (right panel). The black lines show the 10th, 50th, and 90th percentiles under both $\P$ (solid) and $\Q^*$ (dashed). The horizontal dashed grey line corresponds to the stressed VaR, which is approximately 20. Comparing \cref{fig:CVaR_paths} with \cref{fig:VaR_paths} (sample paths of $X$ for the VaR constraint only), we observe that the paths evolve similarly. However, in \cref{fig:CVaR_paths} due to the additional CVaR constraint, if the path crosses $q$ the stressed intensity no longer falls to the baseline parameter, $\kappa = 5$, but instead becomes constant equal to $\kappa \left(b / (b + \eta^*_2)\right)^a = 5.87$. A case in point is the magenta line, where the sample path crosses $q$ and the intensity process becomes constant equal to $5.87$. 

Finally, we turn to a risk-forecasting exercise that can be undertaken in this dynamic framework. At any time prior to the stress, $T$, a risk manager may wish to forecast the expected losses in the next period. In particular, suppose they calculate risk measures of the increment $X_{t+\Delta_t}-X_t$, conditioning on the observed value $X_{t-}$, for $t\leq T- \Delta_t$. 
As the dynamics of the stressed loss process are known in our framework, we can forecast the future losses 
at all times. 
\cref{fig:incr_riskmes} shows the expected value, VaR at the 90\% level, and CVaR at the 90\% level of the quarter-year increment $X_{t+0.25}-X_t$, for times in $[0,0.75]$. Three scenarios are shown under the stressed measure, which correspond to the paths in \cref{fig:CVaR_paths} of the same colour. Clearly, under the baseline the expected value, VaR and CVaR are constant and thus independent of time and state. Under the stressed measure, however, the risk forecasts all depend on time and state. Moreover, the jumps in the underlying paths in \cref{fig:CVaR_paths} lead to sharp increases in the risk forecasts in \cref{fig:incr_riskmes}. Furthermore, for the paths with fewer losses (blue and red), the risk measures tend to decrease when the underlying paths do not jump. However, the risk measures are increasing for the riskier scenario (purple). This risk forecasting exercise illustrates that the proposed dynamic framework leads to richer information for the risk manager, that cannot be obtained in a static setting.

\begin{figure}[tb!]
    \centering
    \includegraphics[width=6in]{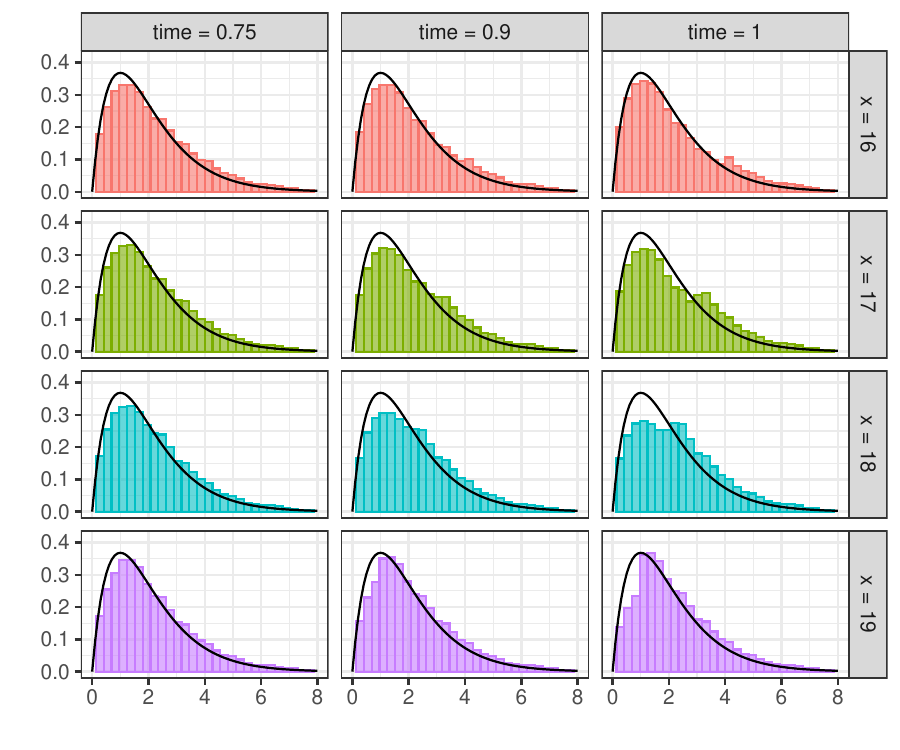}
    \caption{Histogram of 1 million draws from the stressed severity distribution $G^*(t,x,dy)$ under a 15\% increase in $\text{VaR}_{0.9}(X_T)$ and a 12\% increase in $\text{CVaR}_{0.9}(X_T)$ for $t =$ 0.75 (left column), $t = 0.9$ (middle column), and $t=1$ (right column), and $x =$ 16, 17, 18, and 19 (top to bottom). The black lines are the density of the severity random variable under $\P$, which is $\Gamma(2,1)$.}
    \label{fig:CVaR_G}
\end{figure}

\begin{figure}[!]
\centering
\subfloat[Paths $X_t, \, t \in {[0,1]}$ under $\Q^*$]{\includegraphics{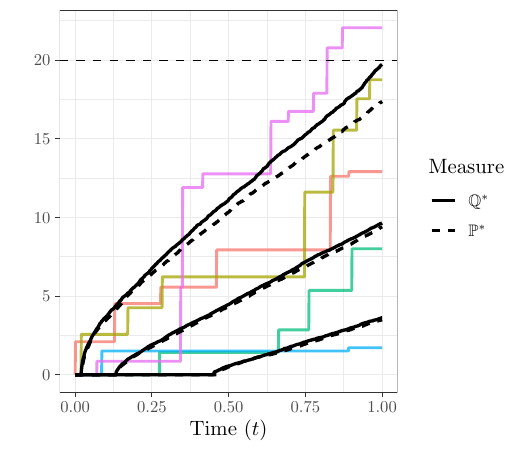}}
\subfloat[Stressed intensity process $\kappa^*(t,x)$]{\includegraphics{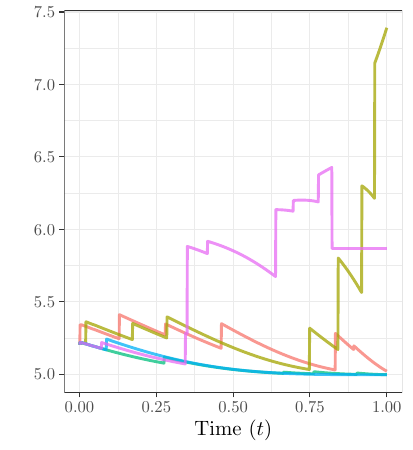}}
\caption{Panel (a): sample paths of $X$ under $\Q^*$. 10,000 paths were simulated and five are highlighted along with the 10\%, 50\%, and 90\% quantiles (black lines) for $\P$ and $\Q^*$. Panel (b): corresponding (same colour) stressed intensity processes. Under $\P$, the severity  distribution is $\xi \sim \Gamma(2,1)$ and the intensity $\kappa=5$. The stresses correspond to a 15\% increase in $\text{VaR}_{0.9}(X_T)$ and a 12\% increase in $\text{CVaR}_{0.9}(X_T)$.}
\label{fig:CVaR_paths}
\end{figure}

\begin{figure}[!]
    \centering
    \includegraphics{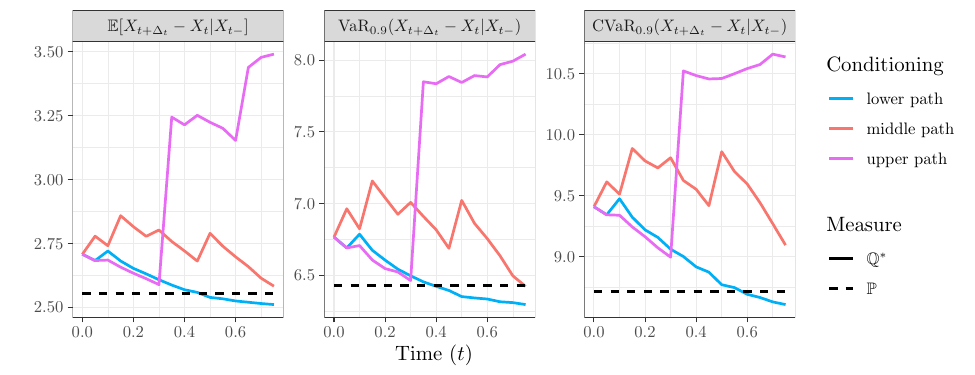}
    \caption{The expected value, VaR at the 90\% level, and CVaR at the 90\% level of the increment $X_{t+\Delta_t}-X_t$, conditioning on the event that $X_t$ is one of the sample paths in \cref{fig:CVaR_paths}. The colours indicate which sample path is used for the conditioning. The forecast is made in time increments of 0.05, forecasting forward $\Delta_t=0.25$. The value of these risk measures under the baseline model is given by the dashed black line.}
    \label{fig:incr_riskmes}
\end{figure}

\end{example}

\section{Aggregate Portfolios and   ``What-if'' Scenarios}
\label{sec:extensions}
Here we provide two extensions of our framework which allow for sensitivity testing of insurance portfolios. In \cref{sec:early}, we solve the optimization problem when a stress is imposed at an earlier time than the terminal time $T$ and when we have two constraints at different times. In \cref{sec:multivariate}, we generalize the results to multivariate compound Poisson processes. Combining these extension allows for application of ``what-if'' scenarios in the spirit of \cref{ex:what-if}. In that application, we consider an aggregate insurance portfolio and study how a stress at time $T^\dagger <T$ on one of the sub-portfolios cascades through time and affects the aggregate portfolio. We observe that a stress on one sub-portfolio perturbs not only the other sub-portfolios but also the dependence between the sub-portfolios.

\subsection{Stresses at Earlier Times}
\label{sec:early}

First, suppose that instead of imposing constraints at the terminal time $T$, we impose constraints at an earlier time $T^\dagger \in (0,T)$. This leads to the following optimization problem:

\begin{optimization}
\label{opti:early-constr}
Let $0 <  T^\dagger  < T <  \infty$,  $f_i \colon \R \to \R$, and $c_i \in \R$ for $i \in [n]$. Consider
\begin{equation*}
    \inf_{\Q\in\Qset} \KL{\Q} \quad \text{s.t.}
    \quad \E^\Q\left[f_i(X_{T^\dagger})\right]=c_i, \quad i \in [n]\,,
\end{equation*}
where $\Qset$ is the class of equivalent probability measures given by
\begin{equation*}
\mathcal{Q}:=\left\{\Q_h\; \Big| \; \frac{d\Q_h}{d\P}=\mathfrak{E}\left(\int_0^{T} \int_\R \left[ h_t(y) - 1 \right] \, \tilde \mu(dy,dt)\right) \right\}
\end{equation*}
and $h: \R_+ \times \R \to \R$ is a predictable, non-negative process satisfying Novikov's condition on $[0,T]$.
\end{optimization}

Note that \cref{opti:early-constr} is not a special case of \cref{opti:main} as the probability measures we seek over are induced by random fields $h$ that live on $[0,T]$. The solution to this new optimization problem is as follows:

\begin{proposition}
If there exists \(\boldeta^* = (\eta_1^*, \ldots, \eta_n^*) \in \R^n\) such that \(\E[\exp \left( -\sum_{i=1}^n \eta_i^* \, f_i(X_{T^\dagger}) \right)] < \infty\) and
\[0 =  \E \left[ \exp \left( -\sum_{j=1}^n\eta^*_j\,f_j(X_{T^\dagger}) \right)\left( f_i(X_{T^\dagger})-c_i \right)\right], \quad i \in [n],\]
then \cref{opti:early-constr} has a solution.

The optimal measure \(\Q^*\) is characterized by the Girsanov kernel
\begin{equation*}
h^*(t,x,y) =  
\begin{cases}
\frac{\E_{t,x+y}\left[\exp \left( -\sum_{i=1}^n\eta^*_i\,f_i(X_{T^\dagger}) \right)\right]}{\E_{t,x}\left[\exp \left( -\sum_{i=1}^n\eta^*_i\,f_i(X_{T^\dagger}) \right)\right]} & \text{if}  \quad t \leq T^\dagger\,, \\
1 & \text{if}  \quad t > T^\dagger \, 
\end{cases}
\end{equation*}
and the corresponding RN derivative is given by
\[\frac{d\Q^*}{d\P} = \frac{\exp \left( -\sum_{i=1}^n\eta^*_i\,f_i(X_{T^\dagger}) \right)}{\E[\exp \left( -\sum_{i=1}^n\eta^*_i\,f_i(X_{T^\dagger}) \right)]}.\]

The solution is unique.
\end{proposition}

\begin{proof}
First, we show that we can restrict to probability measures that are induced by a specific random field. For this let $\Q_h \in \Qset$ for a random field $h$ and define 
$\tilde{h}_t :=h_t \Id_{t \in [0, T^\dagger]} + \Id_{t \in ( T^\dagger, T]}$ and moreover, denote its induced probability measure by $\Q_{\tilde{h}}\in \Qset$. The KL-divergence then becomes
\begin{align*}
D_{KL}(\Q_h \Vert \P) &=  \E^{\Q_h}\left[  \int_0^T\int_\R \left[1-h_t(y)(1-\log h_t(y)) \right]\,\nu(dy,dt)  \right] \\
& \geq \E^{\Q_h}\left[\int_0^{T^\dagger}\int_\R \left[1-h_t(y)(1-\log h_t(y)) \right]\,\nu(dy,dt) \right]  , \quad (\star)
\end{align*}
where the last inequality follows since $h_t\ge 0$ and $\ell(x) :=1-x(1-\log x)$ is a non-negative convex function that attains its minimum when $x \equiv 1$ with $\ell(1) = 0$. Furthermore,
\begin{align*} 
(\star) 
    &=\E \Bigg[  \E \left[ \frac{d\Q_h}{d\P}\Big| \F_{T^\dagger} \right] \int_0^{T^\dagger} \! \! \! \int_\R \left[1-h_t(y)(1-\log h_t(y)) \right]\,\nu(dy,dt)  \Bigg] \\
      &=\E \Bigg[  \exp \left(\int_0^{T^\dagger}\! \! \!\int_\R \left(h_t(y)-1\right)\,\nu(dy,dt) + \int_0^{T^\dagger}\! \! \!\int_\R\log h_t(y)\;\mu(dy,dt) \right)  \int_0^{T^\dagger} \! \! \!\int_\R \left[1-h_t(y)(1-\log h_t(y)) \right]\,\nu(dy,dt)  \Bigg] \\
      &=\E \Bigg[  \exp \left(\int_0^{T} \! \int_\R \left(\tilde h_t(y)-1\right)\,\nu(dy,dt) + \int_0^{T} \! \int_\R\log \tilde h_t(y)\;\mu(dy,dt) \right) \int_0^{T} \! \int_\R \left[1-\tilde h_t(y)(1-\log \tilde h_t(y)) \right]\,\nu(dy,dt)  \Bigg] \\
        &= \E \left[ \frac{d \Q_{\tilde{h}}}{d\P} \log \left( \frac{d \Q_{\tilde{h}}}{d\P} \right) \right] \\
        &= D_{KL}(\Q_{\tilde{h}} \Vert \P) \,.
\end{align*}
Thus we can indeed restrict to probability measures that are induced by random fields that are constant equal to 1 on $(T^\dagger, T]$, i.e., to
\begin{equation*}
    \tilde\Qset :=\left\{\Q_h \in \Qset \; | \; h_t(y) = 1 \text{ for } t \in ( T^\dagger, T] \right\}\,.
\end{equation*}
Next, we observe that this set of probability measures is equivalent to the following set 
\begin{equation*}
    \tilde{\Qset} = 
    \left\{\Q_h \in \Qset \; \Big| \; \frac{d\Q_h}{d\P}=\mathfrak{E}\left(\int_0^{T^\dagger} \int_\R \left[ h_t(y) - 1 \right] \, \tilde \mu(dy,dt)\right) \right\} =:\Qset^\dagger \subset \Qset \,,
\end{equation*}
that is probability measures induced by random fields defined on $[0,T^\dagger]$.
Thus, we showed that in  \cref{opti:early-constr} we can restrict to probability measures in $\Qset^\dagger$. Restricting to $\Qset^\dagger$ provides an optimization problem that is of the form \cref{opti:main} and whose solution and uniqueness is given by \cref{thm:h}, \cref{thm:optim_RN}, and \cref{prop:optim_eta_gen}.
\end{proof}

Next, we impose constraints at both the terminal time $T$ and an earlier time $T^\dagger \in (0,T)$, leading to the following optimization problem:

\begin{optimization}
\label{opti:early-and-terminal-constr}
Let $0 <  T^\dagger  < T <  \infty$,  $f, f^\dagger \colon \R \to \R$, and $c, c^\dagger \in \R$. Consider
\begin{equation*}
    \inf_{\Q\in\Qset} \KL{\Q} \quad \text{s.t.}
    \quad \E^\Q\left[f^\dagger(X_{T^\dagger})\right]=c^\dagger \text{ and } \; \; \E^\Q\left[f(X_{T})\right]=c \,,
\end{equation*}
where $\Qset$ is the class of equivalent probability measures given by
\begin{equation*}
\mathcal{Q}:=\left\{\Q_h\; \Big| \; \frac{d\Q_h}{d\P}=\mathfrak{E}\left(\int_0^{T} \int_\R \left[ h_t(y) - 1 \right] \, \tilde \mu(dy,dt)\right) \right\}
\end{equation*}
and $h: \R_+ \times \R \to \R$ is a predictable, non-negative process satisfying Novikov's condition on $[0,T]$.
\end{optimization}

\begin{proposition}
\label{prop:two-constraints}
If a solution to \cref{opti:early-and-terminal-constr} exists, it is given by $\Q_h$ with $h$ characterized by $h_t(y)=h^{\eta, \eta^\dagger}(t, X_{t^{-}}, y)$, where
\begin{equation*}
        h^{\eta,\eta^\dagger}(t,x,y) = \frac{\E_{t,x+y}\left[  \exp\left( -\eta ( f(X_{T}) -c ) - \eta^\dagger ( f^\dagger(X_{T^\dagger}) -c^\dagger )\Id_{t < T^\dagger} \right) \right]}{\E_{t,x}\left[  \exp\left( -\eta ( f(X_{T}) -c ) - \eta^\dagger ( f^\dagger(X_{T^\dagger}) -c^\dagger )\Id_{t < T^\dagger} \right) \right]} \, 
    \end{equation*}
and $\eta,\eta^\dagger$ are Lagrange multipliers such that the constraints hold. Furthermore, the solution, if it exists, is unique.
\end{proposition}

\begin{proof}
    The Langrangian associated with the problem is 
    \begin{equation*}
        \E^{\Q_h} \left[ \log \left( \frac{d \Q_h}{d\P} \right) + \eta^\dagger (f^\dagger(X_{T^\dagger}) -c^\dagger) + \eta ( f(X_{T}) -c ) \right]\,.
    \end{equation*}
    As before, we write
    \begin{equation*}
        \E^{\Q_h} \left[ \log \left( \frac{d \Q_h}{d\P} \right) \right] = \E^{\Q_h} \left[ \int_0^T \int_\R \big(1-(1-\log h(t,y))\;h(t,y) \big)\,\kappa \, G(dy) dt  \right].
    \end{equation*}
    Thus, the associated value function at time $t$ is
    \begin{equation*}
        J^{\eta,\eta^\dagger}(t,x) = \inf_{\Q\in\Qset} \E^\Q_{t,x}\left[ \int_t^T\int_\R \big(1-(1-\log h(t,y))\;h(t,y) \big)\,\kappa \, G(dy) dt + \eta^\dagger (f^\dagger(X_{T^\dagger}) -c^\dagger) + \eta ( f(X_{T}) -c ) \right].
    \end{equation*}
    We consider two cases: First, if $t \in (T^\dagger, T]$, then $f^\dagger (X_{T^\dagger})$ is $\mathcal{F}_t$-measurable, so
        \begin{equation*}
        J^{\eta,\eta^\dagger}(t,x) = \eta^\dagger (f^\dagger(X_{T^\dagger}) -c^\dagger) + \ell^\eta(t,x) \,,
    \end{equation*}
    where 
    \begin{equation*}
    \ell^\eta(t,x) := \inf_{\Q\in\Qset} \E^\Q_{t,x}\left[ \int_t^T\int_\R \big(1-(1-\log h(t,y))\;h(t,y) \big)\,\kappa \, G(dy) dt + \eta ( f(X_{T}) -c ) \right], \quad t \in (T^\dagger, T].
    \end{equation*}
    This is a value function of the same form as in the proof of \cref{thm:h}. Thus by the same procedure, we find that the optimal control is $h^\eta(t,x,y) = e^{-\Delta_y \ell^\eta(t,x)}$ and 
    \begin{equation}
    \label{eq:ell-eta}
        \ell^\eta(t,x)  = - \log \left( \E_{t,x}\left[ e^{-\eta ( f(X_{T}) -c )} \right]\right), \quad t \in (T^\dagger, T].
    \end{equation}
    Second, if $t \in (0, T^\dagger]$, then we write the value function as follows
        \begin{equation*}
        J^{\eta,\eta^\dagger}(t,x) = \inf_{\Q\in\Qset} \E^\Q_{t,x}\left[ \int_t^{T^\dagger}\int_\R \big(1-(1-\log h(t,y))\;h(t,y) \big)\,\kappa \, G(dy) dt + \eta^\dagger (f^\dagger(X_{T^\dagger}) -c^\dagger) + \ell^\eta(T^{{\dagger+}},x) \right] \,,
    \end{equation*}
    where $\ell^\eta(T^{{\dagger+}},x): = \lim_{t\searrow T^\dagger}\ell^\eta(t,x)$ denotes the limit of $\ell^\eta(t,x)$ as $t$ approaches $T^\dagger$ from the right.
    $J^{\eta,\eta^\dagger}(t,x) $ satisfies the HJB equation
    \begin{align*}
    \partial_t J^{\eta,\eta^\dagger}(t,x)+ \inf_{h} \left\{ \L^{h} J^{\eta,\eta^\dagger}(t,x) + \int_\R \big(1-(1-\log h(t,y))\;h(t,y) \big) \, \kappa \, G(dy)  \right\} &=0 ,
    \\
    J^{\eta,\eta^\dagger}(T^\dagger,x) &=  \eta^\dagger (f^\dagger(x) -c^\dagger) + \ell^\eta(T^{\dagger+},x) .
\end{align*}
    The optimal control is $h^\eta(t,x,y) = e^{-\Delta_y J^{\eta,\eta^\dagger}(t,x)}$, and again using the Cole-Hopf change of variables and the Feynman-Kac Theorem, we obtain
    \begin{equation*}
        J^{\eta,\eta^\dagger}(t,x)  = - \log \left( \E_{t,x}\left[  \exp\left(-\eta^\dagger ( f^\dagger(X_{T^\dagger}) -c^\dagger ) + \ell^\eta(T^{\dagger+},X_{T^\dagger}) \right) \right]\right), \quad t \in (0, T^\dagger].
    \end{equation*}
    Substituting in \eqref{eq:ell-eta} and using the tower property of expectations, we obtain
    \begin{equation*}
        J^{\eta,\eta^\dagger}(t,x)  = - \log \left( \E_{t,x}\left[  \exp\left(-\eta^\dagger ( f^\dagger(X_{T^\dagger}) -c^\dagger ) -\eta ( f(X_{T}) -c ) \right) \right]\right), \quad t \in (0, T^\dagger].
    \end{equation*}
    Taken together, we have
    \begin{equation*}
        h^{\eta,\eta^\dagger}(t,x,y) = \frac{\E_{t,x+y}\left[  \exp\left( -\eta ( f(X_{T}) -c ) - \eta^\dagger ( f^\dagger(X_{T^\dagger}) -c^\dagger )\Id_{t < T^\dagger} \right) \right]}{\E_{t,x}\left[  \exp\left( -\eta ( f(X_{T}) -c ) - \eta^\dagger ( f^\dagger(X_{T^\dagger}) -c^\dagger )\Id_{t < T^\dagger} \right) \right]} \, ,
    \end{equation*}
    and 
    \begin{equation*}
        J^{\eta,\eta^\dagger}(t,x) = \begin{cases}
            \eta^\dagger (f^\dagger(X_{T^\dagger}) -c^\dagger) - \log \left( \E_{t,x}\left[ e^{-\eta ( f(X_{T}) -c )} \right]\right) & t < T^\dagger,\\
            - \log \left( \E_{t,x}\left[  e^{-\eta^\dagger ( f^\dagger(X_{T^\dagger}) -c^\dagger ) -\eta ( f(X_{T}) -c ) }\right]\right) & t \geq T^\dagger.
        \end{cases}
    \end{equation*}
\end{proof}

To obtain expressions for the optimal Lagrange multipliers, we note that \cref{thm:special-stoch-exp} can be extended to show that if $\E \left[ \exp \left(-\eta ( f(X_{T}) -c )  - \eta^\dagger ( f^\dagger(X_{T^\dagger}) -c^\dagger )\right) \right] < \infty$, then the RN density of the solution $\Q_h$ from \cref{prop:two-constraints} admits the form
\begin{equation}
\label{eq:two-time-RN}
    \frac{d\Q^{\eta,\eta^\dagger}}{d\P} = \frac{\exp\left( -\eta ( f(X_{T}) -c ) \right) - \eta^\dagger ( f^\dagger(X_{T^\dagger}) -c^\dagger )}{\E\left[  \exp\left( -\eta ( f(X_{T}) -c ) \right) - \eta^\dagger ( f^\dagger(X_{T^\dagger}) -c^\dagger ) \right]} \,.
\end{equation}
This representation gives us the following system of equations for $\eta$ and $\eta^\dagger$:

\begin{lemma}
There exists a unique solution to \cref{opti:early-and-terminal-constr} if there exist Lagrange multipliers $\eta^\dagger, \eta$ such that $ \E \left[ \exp \left(-\eta ( f(X_{T}) -c )  - \eta^\dagger ( f^\dagger(X_{T^\dagger}) -c^\dagger )\right) \right] < \infty$ and the system of equations
\begin{align*}
    c^\dagger &= \E\left[ \exp\left( -\eta ( f(X_{T}) -c ) - \eta^\dagger ( f^\dagger(X_{T^\dagger}) -c^\dagger ) \right)  (f^\dagger(X_{T^\dagger}) -c^\dagger )\right] ,\\
    c &= \E\left[ \exp\left( -\eta ( f(X_{T}) -c )  - \eta^\dagger ( f^\dagger(X_{T^\dagger}) -c^\dagger )  \right) (f(X_{T}) -c)\right],
\end{align*}
has a solution.
\end{lemma}
\begin{proof}
Rewriting the constraints using \eqref{eq:two-time-RN} concludes the proof.
\end{proof}

Observe that the proof of \cref{prop:two-constraints} can be generalized to a finite number of constraints at different time points, by splitting the value function into additional pieces. An alternative approach to obtain the stressed measure and the dynamics of the process under the stressed measure, is to increase the state space and introduce stopped processes. Specifically, suppose that we have a constraint of the form $\E^\Q[f(X_{T_1},\dots,X_{T_n})]=c$ where $0<T_1<\dots<T_n\le T$, then, by introducing the set of stopped processes $X^{(i)}_t:=X_t\mathds{1}_{\{t\le T_i\}} + X_{T_i}\mathds{1}_{\{t>T_1\}}$ for $i=1,\dots,n$, the original constraint can be viewed as a single constraint on a multi-variate process at terminal time, i.e., $\E^\Q[f(X_{T}^{(1)},\ldots, X_{T}^{(n)})]=c$. Thus the approach in the next section, where we consider multivariate compound Poisson processes, can be applied to this new problem.

\subsection{Multivariate Compound Poisson Processes}\label{sec:multivariate}
In this section, we consider a $d$-dimensional compound Poisson process $\boldsymbol{X} := \left(X_t^1, \ldots, X_t^d\right)_{t \in [0,T]}$. We denote (random) vectors using bold font to distinguish them from (random) variables and scalars. We assume that, under $\P$, $\boldsymbol{X}$ has mean measure
\begin{equation*}
    \nu(\boldsymbol{x},dt) = \kappa \, G(d\boldsymbol{x}) dt,
\end{equation*} 
where $G$ is the $d$-dimensional severity distribution and $\kappa>0$ the scalar intensity. The optimization problem we consider in this section is a stress applied to one of the components of $\boldsymbol X$ at time $T^\dagger \in (0,T]$. For simplicity of notation, we stress throughout the first component of $\boldsymbol{X}$, i.e., $X^1$.

\begin{optimization}
\label{opti:x1}
Let $0 <  T^\dagger  \leq T <  \infty$, $f_i \colon \R \to \R$ and $c_i \in \R$ for $i \in [n]$, and consider
\begin{equation*}
    \inf_{\Q\in\Qset} D_{KL}(\mathbb{Q} \, \Vert \, \mathbb{P} )
    \quad
    \text{s.t.}
    \quad
    \E^\Q\left[f_i(X^1_{T^\dagger})\right]=c_i, \quad i \in [n]\,,
\label{eq:gen_prob_mult}
\end{equation*}
where $\Qset$ is the class of equivalent probability measures induced by Girsanov's theorem
\begin{equation*}
\label{eq:Q_multiv}
\mathcal{Q}:=\left\{\Q_h\; \Big| \; \frac{d\Q_h}{d\P}=\mathfrak{E}\left(\int_0^T \int_{\R^d} \left[ h_t(\boldsymbol{y}) - 1 \right] \, \tilde \mu(d\boldsymbol{y},dt)\right) \right\} ,
\end{equation*}
and $h: \R_+ \times \R^d \to \R$ is a predictable, non-negative random field satisfying Novikov's condition on $[0,T]$.
\end{optimization}

The solution to \cref{opti:x1} is given in the next proposition.
\begin{proposition}
\label{prop:h_RN_biv}
If there exists $\boldeta^* = (\eta_1^*, \ldots, \eta_n^*) \in \R^n$ such that $\E[\exp \left( -\sum_{i=1}^n \eta_i^* \, f_i(X^1_{T^\dagger}) \right)] < \infty$ and
\begin{equation*}
0 =  \E \left[ \exp \left( -\sum_{j=1}^n\eta^*_j\,f_j(X_{T^\dagger}^1) \right)\left( f_i(X_{T^\dagger}^1)-c_i \right)\right] \quad \text{for } i \in [n],
\end{equation*}
then \cref{opti:x1} has a solution. The solution is the measure $\Q^*$ characterized by the Girsanov kernel
\begin{equation*}
h^*(t,\boldsymbol{x},\boldsymbol{y}) =  
\begin{cases}
\frac{\E_{t,\boldsymbol{x}+\boldsymbol{y}}\left[\exp \left( -\sum_{i=1}^n\eta^*_i\,f_i(X^1_{T^\dagger}) \right)\right]}{\E_{t,\boldsymbol{x}}\left[\exp \left( -\sum_{i=1}^n\eta^*_i\,f_i(X^1_{T^\dagger}) \right)\right]} & \text{if}  \quad t \leq T^\dagger \\
1 & \text{if}  \quad t > T^\dagger\,.
\end{cases}
\end{equation*}
The solution is unique. 
\end{proposition}

We omit the proof of \cref{prop:h_RN_biv} as it follows using steps similar to those in the proof of \cref{thm:h} and arguments in \cref{sec:early}. Note that if a stress is applied at the terminal time, i.e., $T^\dagger = T$, then $h^*(t,\boldsymbol{x},\boldsymbol{y})$ simplifies to a representation akin to that in \cref{sec:main}. Next, we show that the optimal Girsanov kernel is a function of $x_1$ and $y_1$ only.   
\begin{corollary}
The optimal Girsanov kernel depends only on the first components of $\boldsymbol{x}$ and $\boldsymbol{y}$, i.e., 
\begin{align*}
    h^*(t,\boldsymbol{x},\boldsymbol{y}) &=
    \frac{\left.\E\left[\exp \left( -\sum_{i=1}^n\eta^*_i\,f_i(X^1_{T^\dagger}) \right)\, \right|\, X^{1}_{t^-} = x_1+y_1\right]}
    {\left.\E \left[\exp \left( -\sum_{i=1}^n\eta_i^*\,f_i(X^1_{T^\dagger}) \right)\, \right|\, X^1_{t^-} =x_1\right]} 
    =: h^*(t,x_1,y_1).
\end{align*}
\end{corollary}
This follows since under the reference measure $\boldsymbol{X}$ has independent increment and the severity distributions do not depend on states. 

Next, we examine how the dynamics of $\boldsymbol{X}$ change when moving from $\P$ to a stressed measure $\Q^*$. Specifically, we are interested in how the other components of the process, $X^j$ for $j = 2, \ldots, d$ are affected by a stress on $X^1$. For this, we denote the marginal severity random variables by $\xi_i$, $i \in [d]: = \{1, \ldots, d\}$ and their distributions under the reference measure by $G_i(z) := \P(\xi_i \le z)$. For a stressed distribution $\Q^*$ we write  $G_i^*(t,\boldsymbol{x},B) := \Q^*(\xi_i \in B | \boldsymbol{X}_{t-} = \boldsymbol{x})$ for $t\in [0,T]$ and $B \in \B(\R)$.

\begin{proposition}
\label{prop:kappa_G_biv}
The stressed intensity and the stress joint severity distribution (solution to \cref{opti:x1}) are, respectively,
\begin{equation*}
    \kappa^*(t,x_1) =  \kappa \int_\R h^*(t,x_1,y_1) \, G_1(dy_1) \quad \text{and} \quad
    G^*(t,x_1,d\boldsymbol{y}) = \frac{h^*(t,x_1,y_1) G(d\boldsymbol{y})}{\int_\R  h^*(t,x_1,y_1') G_1(dy_1')},
\end{equation*}
where we set $\kappa^*(t,x_1) := \kappa^*(t,\boldsymbol{x})$ and  $G^*(t,x_1,d\boldsymbol{y}):= G^*(t,\boldsymbol{x},d\boldsymbol{y})$, since the stressed intensity and the stressed severity distribution only depend on the value of the stressed component $x_1$.
\end{proposition}
\begin{proof}
By Girsanov's theorem, the compensator of $\mu(d\boldsymbol{y},dt)$ under the stressed measure is $\nu^*(d\boldsymbol{y},dt)=h^*(t,x_1,y_1)\,\nu(d\boldsymbol{y},dt)$. The result follows analogous to the proof of \cref{prop:kappa_G}.
\end{proof}

Notice that under the stressed measure all components $X^i$, $i \in [d]$, share the same stressed intensity process $\kappa^*(t, x_1)$ given in \cref{prop:kappa_G_biv} and thus the dynamics of all components are distorted. We report the stressed marginal severity distributions $G_i^*$, $i \in [d]$, in the next proposition.

\begin{proposition}\label{prop:G_biv}
The stressed marginal severity distributions of $\boldsymbol{X}$ under $\Q^*$ (solution to \cref{opti:x1}) are
\begin{equation*}
    G^*_1(t,x_1, dy_1) = \frac{h^*(t,x_1,y_1) G_1(dy_1)}{\int_\R  h^*(t,x_1,y_1') G_1(dy_1')}
    \quad \text{and}
    \quad G^*_j(t,x_1,dy_j) = \frac{\int_\R h^*(t,x_1,y_1)  G_{1,j}(dy_1,dy_j)}{\int_\R  h^*(t,x_1,y_1') G_1(dy_1')} 
\end{equation*}
for $j \in [d] \setminus \{1\}$ and where $ G_{1,j}(dy_1,dy_j)$ is the joint severity distribution of the first and the $j^\text{th}$ components under $\P$.
\end{proposition}
\begin{proof}
The result follows from \cref{prop:kappa_G_biv} by integration.
\end{proof}
We note that the stressed marginal severity distribution of the first component of $\boldsymbol{X}$, $G^*_1$, is equal to the stressed severity distribution in the 1-dimensional case (\cref{prop:kappa_G}). The $j$-th marginal severity distribution $G^*_j$, $j \neq 1$, however, depends on the copula between $(\xi_1, \xi_j)$. If $\xi_j$ is independent from $\xi_1$, then the stressed marginal severity distribution of $\xi_j$ is equal to its distribution under $\P$, i.e., $G^*_j(t,x_1,dy_j) = G_j(d y_j)$. If $\xi_1$ and $\xi_j$ are dependent under $\P$, a stress on $X^1$ leads to a stressed marginal severity distributions of $\xi_j$ which may depend on time $t$ and state $x_1$. We illustrate in \cref{ex:copulas} how different copulas between $(\xi_1, \xi_j)$ alter the stressed severity distributions.

Even if the severity distributions are independent under $\P$, implying that $G_j^* = G_j$, $j\neq 1$, the paths of the components of $\boldsymbol{X}$ under $\Q^*$ are in general still dependent. This is due to the shared intensity process $\kappa^*(t,x_1)$. In the next example, we construct a joint severity distribution of $(\xi_1, \xi_2)$ under $\P$ such that the paths of $X^1$ and $X^2$ are independent under $\P$ and $\Q^*$.

\begin{example}
\label{ex:ind_mix}
We consider a bivariate compound Poisson process under $\P$ with a specific dependence between the severity distributions, which we will refer to as the ``independent mixture'' distribution. Let $G_a$ and $G_b$ be distribution functions and let $p \in (0,1)$. Suppose the severity distributions $(\xi_1,\xi_2)$ are such that
\begin{equation*}
G_{1,2}(d y_1, d y_2)
=
p \, \delta_{0} \, (y_2) G_a(dy_1) + (1-p) \, \delta_0(y_1) \, G_b(d y_2),
\end{equation*}
where $\delta_0$ denotes the Dirac measure at $0$. This choice implies that under $\P$ if the process $\boldsymbol{X}$ jumps, only $X^1$ or $X^2$ will jump. In particular, the parameter $p$ (resp. $1-p$) is the probability that $X^1$ (resp. $X^2$) jumps, conditional on the event that the process jumps. 

Under the stressed measure $\Q^*$ (solution to \cref{opti:x1}) the stressed intensity of the process is, applying \cref{prop:kappa_G_biv},
\begin{equation*}
    \kappa^*(t,x_1) = 
    \kappa \left[1 - p + p \int_\R h^*(t,x_1,y_1) G_a(dy_1) \right]\,,
\end{equation*}
as $h^*(t,x_1,0) = 1$. The joint stressed severity distribution  becomes
\begin{align*}
    G^*(t,x_1,d\boldsymbol{y})
    &=  p^*(t,x_1) \delta_0(y_2)\,\frac{h^*(t,x_1,y_1)  G_a(dy_1)}{\int_\R  h^*(t,x,dy_1^\prime) G_a(dy_1^\prime)}  
    +
    [1-p^*(t,x_1)] \delta_0(y_1) \, h^*(t,x_1,y_1) G_b(dy_2) ,
\end{align*} 
where 
\begin{equation*}
    p^*(t,x_1) 
    :=
    \frac{\kappa}{\kappa^*(t, x_1) }\,p \int_\R h^*(t,x_1,y_1^\prime) G_a(dy_1^\prime) \qquad \text{and}\qquad
    1 - p^*(t,x_1) 
    :=
    \frac{\kappa}{\kappa^*(t, x_1)} (1-p)\,.
\end{equation*}
Thus, under $\Q^*$ if $\boldsymbol{X}$ jumps, only $X^1$ or $X^2$ jump. Moreover, the probability that $X^1$ (resp. $X^2$) jumps, given that the process jumps, is $p^*(t, x_1)$ (resp. $1-p^*(t, x_1)$).

Next, we show that the probability that $X^2$ jumps is identical under both $\P$ and $\Q^*$. For this, let $N$ denote a Poisson process with intensity $\kappa$ under $\P$ and $N^*$ a Poisson process with intensity $\kappa^*(t,x_1)$ under $\Q^*$ such that $(X_t^1, X_t^2) = \sum_{i = 1}^{N_t} (\xi_{1,i}, \xi_{2,i})$ under $\P$ and $(X_t^1, X_t^2) = \sum_{i = 1}^{N^*_t} (\xi_{1,i}, \xi_{2,i})$ under $\Q^*$, where for $j = 1,2$,  $\xi_{j,i}$ are i.i.d. with distribution $G_j$ and $G_j^*$ respectively. Note that under $\Q^*$ the process $N^*$ is independent from the corresponding severity random variables. Thus we obtain, for all $t \ge 0$ and $\Delta t>0$ small enough, that
\begin{align*}
    \Q^*(X^2_t - X^2_{t-\Delta t} > 0) &= \Q^*(N^*_{\Delta t} = 1) \, \Q^*(\xi_2 > 0)  + o(\Delta t)\\
    &= \kappa^*(t,x_1) \Delta t \,(1-p^*(t,x_1)) + o(\Delta t)\\
    &= \kappa \Delta t \, (1-p)+ o(\Delta t) \\
    &= \P(N_{\Delta t} = 1) \, \P(\xi_2 > 0)  + o(\Delta t)\\
    &= \P(X^2_t - X^2_{t-\Delta t} > 0) \, .
\end{align*}
Moreover, the marginal severity distribution of $\xi_2$, conditional on the event that $X^2$ jumps, is $G_b(dy_2)$ under both $\P$ and $\Q^*$. Therefore, even though $X^2$ has under $\Q^*$ an intensity process $\kappa^*(t,x_1)$ that depends on $X^1$, the change to the intensity and the change to the probability of the jump occurring to $X^2$ cancel, leaving the process $X^2$ unchanged.
\end{example}

If $X^1$ and $X^2$ are dependent, then the stressed measure affects the dependence between the paths. The next example examines how the dependence structure between $X^1$ and $X^2$ changes for different copulas of the joint severity distribution $(\xi_1, \xi_2)$.

\begin{figure}[!tbh]
\centering
\includegraphics[width=4in]{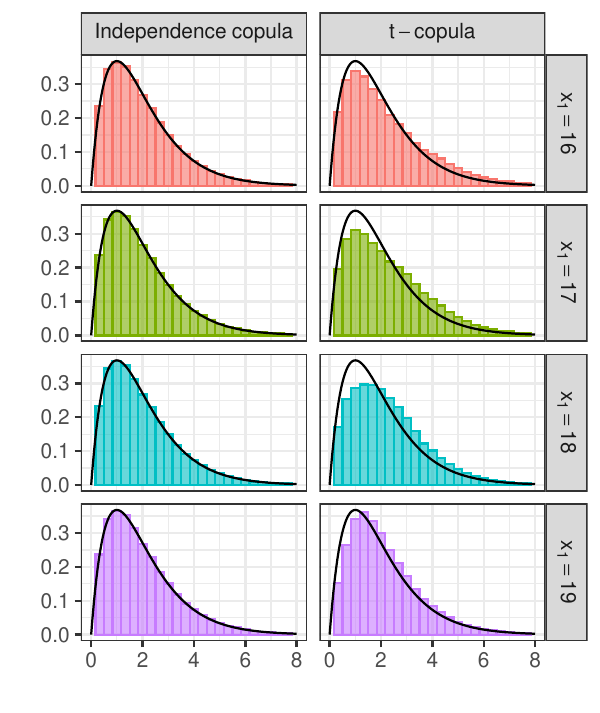}
    \caption{Histogram of 1 million draws from $G_2^*(T,x_1,dy_2)$ at the terminal time $T=1$ for state space values $x_1 =$ 16, 17, 18, and 19, when a stress is applied to the first component of the process only. Under $\P$, the severity random variables $(\xi_1,\xi_2)$ are independent (left column) and have a $t$-copula with 3 degrees of freedom and correlation 0.8 (right column). Under $\P$, the intensity is $\kappa=5$ and the marginal severity distributions are both $\Gamma(2,1)$.}
    \label{fig:x2_jumps_biv}
\end{figure}

\begin{example}
\label{ex:copulas}
Under $\P$, we consider a bivariate compound Poisson process over the time horizon $[0,T]$, $T = 1$, with marginal severity distributions $\xi_1, \xi_2 \sim \Gamma(2, 1)$ and intensity $\kappa = 5$. With these parameters, we have $\text{VaR}_{0.9}(X^1_T) = 17.4$. We apply a 15\% increase in VaR at the 90\% level to the first component, i.e. $\text{VaR}^\Q_{0.9}(X^1_T)=19.97$. \cref{fig:x2_jumps_biv} displays histograms of the stressed severity distribution $\xi_2$ at time $T = 1$ for $x_1 =$ 16, 17, 18, and 19. The panels in the left column display the histograms under the assumption that $(\xi_1,\xi_2)$ are independent under $\P$, while the panels in the right column display the histograms under the assumption that $(\xi_1,\xi_2)$ have under $\P$ a $t$-copula with 3 degrees of freedom and correlation 0.8. The black line is the density of the severity under $\P$, $\Gamma(2,1)$. We observe that when the reference severity distributions are independent, the stressed severity distribution of the second component is unchanged from the reference model. However, when the severity distributions are dependent, the severity distribution of the second component gets distorted under the stressed measure. The distortion is most noticeable for $x_1=17$ and 18. As the second component of the process is only stressed indirectly, this distortion reflects a ``spill-over'' effect of the stress on the first component.

\begin{figure}[!t]
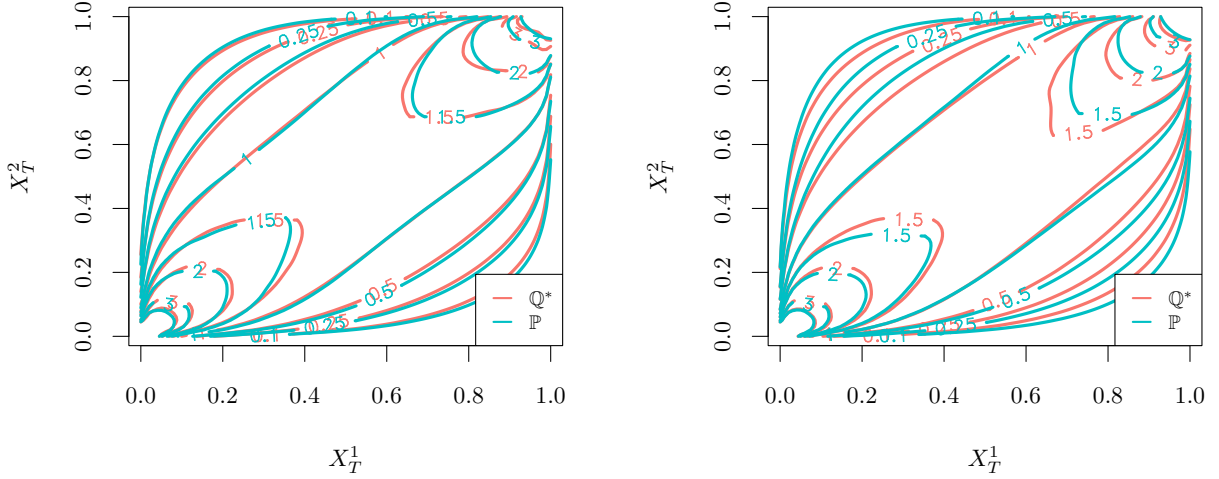

\centering
\subfloat[Copula of $(X_T^1, X_T^2)$ -- independent severity distributions  ]{\input{figures/independence_copula_updated.tex}}
\subfloat[Copula of $(X_T^1, X_T^2)$ -- severity distributions with $t$ copula, correlation 0.8, 3 d.f.]{\input{figures/t_copula_updated.tex}}
\caption{Contour plots of the densities of the empirical copula of $(X_T^1, X_T^2)$ under $\P$ and $\Q^*$, where $\Q^*$ is the optimal measure under a 15\% increase to  $\text{VaR}_{0.9}(X^1_T)$. Under $\P$, $\boldsymbol{X}$ is a compound Poisson process with marginal severity distributions $\xi_1,\xi_2 \sim \Gamma(2, 1)$ and intensity $\kappa = 5$. In panel (a), $\xi_1$ and $\xi_2$ are independent, while in panel (b), they are dependent via a $t$-copula with correlation 0.8 and 3 degrees of freedom.}
    \label{fig:copula}
\end{figure}

\cref{fig:copula} shows contour plots of the copula densities of $(X^1_T, X_T^2)$ for the same two choices of copulas for the severity distribution. We observe an increase in upper tail dependence between $X_T^1$ and $X_T^2$ under $\Q^*$ when the marginal severity distributions $\xi_1$ and $\xi_2$ are dependent. 
\cref{tab:spearman} reports the Spearman's rank correlation coefficient between $X^1_T$ and $X^2_T$ under $\P$ and $\Q^*$, i.e. $\rho_\tau(X_T^1, X_T^2)$ and $\rho_\tau^{\Q^*}(X_T^1, X_T^2)$, for different copulas. The 95\% confidence intervals of the estimates are reported in brackets, computed using the method of \textcite{bonett2000sample}. When the severity distribution of $X^1$ and $X^2$ is the independent mixture from \cref{ex:ind_mix}, $\rho_\tau(X_T^1, X_T^2)$ is approximately 0 under both $\P$ and $\Q^*$. When $X^1$ and $X^2$ have independent severity distributions under $\P$, we observe an increase in correlation under $\Q^*$, however it is not statistically significant at the 5\% level. However, for most other choices of copulas, there is a statistically significant increase in the correlation under $\Q^*$. 
\end{example}

\begin{table}[!t]
    \centering
    \begin{tabular}{ c @{\hskip 3em} r @{\hskip 2em} r }
    Copula of $(\xi_1, \xi_2)$ under $\P$& $\rho_\tau(X_T^1, X_T^2)$ & $\rho_\tau^{\Q^*}(X_T^1, X_T^2)$\\[0.5em]
\toprule
Independent mixture & $-0.003$ ($\pm$0.019) & $-0.010$ ($\pm$0.020)\\[0.5em]
Independence copula & 0.674 ($\pm$0.011) & 0.692 ($\pm$0.012)\\[0.5em]
$t$-copula, corr = 0.8, df = 1 & 0.668 ($\pm$0.011) & 0.707 ($\pm$0.011)\\[0.5em]
$t$-copula, corr = 0.8, df = 3 & 0.661 ($\pm$0.012) & 0.710 ($\pm$0.011)\\[0.5em]
Gumbel copula, $\theta=2$ & 0.672 ($\pm$0.012) & 0.691 ($\pm$0.011)\\[0.5em]
Frank copula, $\theta=2$ & 0.671 ($\pm$0.012) & 0.694 ($\pm$0.012)\\
\bottomrule
    \end{tabular}
    \caption{Spearman's rank correlation coefficient with 95\% confidence bounds in brackets for different copulas of $(\xi_1, \xi_2)$, where $\xi_1, \xi_2 \sim \Gamma(2,1)$ and the intensity is $\kappa = 5$. A stress of a 15\% increase in $\text{VaR}_{0.9}(X^1_T)$ is imposed.}
    \label{tab:spearman}
\end{table} 

Finally, we examine a simulation approach to ``what-if'' scenarios. We consider an insurance loss portfolio where one component of the portfolio ($X^1$) undergoes a stress at time $T/2$ and are interested in how the stress impacts the risk (measured via a risk measure) of the aggregate portfolio. In a risk management context one may be interested in the required severity of a stress on $X^1$ at time $T/2$ that leads to a breach of a risk threshold at the terminal time. 

\begin{example}\label{ex:what-if}
Suppose we have a bivariate process $\boldsymbol{X} = (X^1_t, X^2_t)_{t\in[0,T]}$ and impose a stress on $\text{VaR}_\alpha(X^1_{T/2})$ for some $\alpha \in(0,1)$. We then examine the effect of the stress on the VaR of the aggregate portfolio at the terminal time, i.e., on $\text{VaR}_{0.9}(X^1_T + X^2_T)$. 
Specifically, for $\alpha\in (0,1)$, we calculate the minimal percentage increase in $\text{VaR}_\alpha(X^1)$ such that VaR at the 90\% level of the aggregate portfolio is at least 5\% larger under the stressed measure, i.e. $\text{VaR}_{0.9}^{\Q^*}(X^1_T + X^2_T) \ge 1.05 \text{VaR}_{0.9}(X^1_T + X^2_T)$.

\cref{tab:what_if} reports the corresponding increase in $\text{VaR}_\alpha(X^1_{T/2})$ in percentage for different $\alpha$ levels, where $\xi_1 \stackrel{\P}{\sim} \Gamma(2,1)$, $\xi_2 \stackrel{\P}{\sim}\text{Exp}(2)$, $\kappa=5$, and $(\xi_1, \xi_2)$ have under $\P$ a $t$-copula with 3 degrees of freedom and correlation 0.8. Since the stress increases the positive dependence between $X^1$ and $X^2$, we see that a larger stress for smaller levels of $\alpha$ is required to achieve the same increase in $\text{VaR}_{0.9}(X^1_{T} + X^2_T)$.

\begin{table}[thb]
    \centering
    \begin{tabular}{ c @{\hskip 3em} c }
   $\alpha$ &  Stress: \% increase in $\text{VaR}_\alpha(X_{T/2}^1)$\\
\toprule
0.3 & 58.6\\
0.4 & 40.4\\
0.5 & 30.4\\
0.6 & 24.8\\
0.7 & 19.6\\
0.8 & 16.1\\
\bottomrule
    \end{tabular}
        \caption{Required percentage increase in $\text{VaR}_\alpha(X^1_{T/2})$ for various levels of $\alpha$ to achieve a 5\% increase in $\text{VaR}_{0.9}(X^1_{T} + X^2_T)$. $\boldsymbol{X}$ is a compound Poisson process with intensity $\kappa=5$, marginal severity random variables $\xi_1 \stackrel{\P}{\sim} \Gamma(2,1)$, $\xi_2 \stackrel{\P}{\sim} \text{Exp}(2)$. $(\xi_1,\xi_2)$ have a $t$ copula with correlation 0.8 and 3 degrees of freedom.}
    \label{tab:what_if}
\end{table}
\end{example}

\vspace{3em}
\section{Simulating Under the Stressed Measure}
\label{sec:numerics}
To illustrate how a stress affects the dynamics of a compound Poisson process, we plot in the earlier examples different sample paths of the process, its intensity processes, and the severity distributions under the reference and the stressed probability measure. Key to simulating sample paths under the stressed measure is to (a) find the optimal Lagrange multiplier $\boldeta^*$, and (b) calculate the Girsanov kernel $h^*(t,x,y)$, which then allows to estimate the stressed intensity $\kappa^*(t,x)$ and the stressed severity distribution $G^*(t,x)$ using Proposition \ref{prop:kappa_G}
To find the optimal Lagrange multiplier, we solve \cref{eq:find_eta} by estimating the required expectation using the Fourier space time-stepping (FST) algorithm \citep{jackson2008fourier}. For completeness, we recall the FST methodology in \cref{sec:FST}. To estimate the Girsanov kernel $h^*(t,x,y)$, which is given in \cref{thm:h}, we again utilize the FST for calculating conditional expectations, see \cref{sec:simulation}.

The code implementing the algorithms is available at \href{https://github.com/emmakroell/stressing-dynamic-loss-models}{https://github.com/emmakroell/stressing-dynamic-loss-models}.

\subsection{Fourier space time-stepping algorithm}
\label{sec:FST}
In this section, we provide an algorithm to simulate paths of $X$ under the stressed measure, which is based on the Fourier space time-stepping (FST) algorithm \citep{jackson2008fourier}. FST is  a method for efficiently solving partial integro-differential equations (PIDE) such as \cref{eqn:HJB-linearized} using fast Fourier transform methods. We briefly summarize the FST method before describing our algorithm for simulating under the stressed measure.

For a function $\ell$, we  denote its Fourier transform by $\hat \ell$ and consider a PIDE with a general terminal condition $g(x)$:
\begin{equation}
\label{eq:PIDE_ex}
  {\left \lbrace \begin{aligned}
  \partial_t\omega(t,x) + \L \omega(t,x) &= 0 \, , \\
   \omega(T,x)  &=  g(x)  \, ,
  \end{aligned}\right.}
\end{equation}
where the linear operator $\L$ is  the $\P$-generator of the process $X$. Using Fourier transform, we rewrite the above PIDE as
\begin{equation}
  {\left \lbrace \begin{aligned}
  (\partial_t + \Psi(\zeta)) \hat \omega(t,\zeta) &= 0 \, , \\
    \hat \omega(T, \zeta) &= \hat g(\zeta) \, ,
  \end{aligned}\right.}
  \label{eq:fourier_PDE}
\end{equation}
where $\Psi$ is the characteristic function of $X_T$. As \eqref{eq:fourier_PDE} is a linear ODE, its solution is
\begin{equation}
    \hat \omega (t,\zeta) = \hat g (\zeta)  e^{\Psi(\zeta)(T-t)}.
    \label{eq:fourier_soln}
\end{equation}
Thus a solution to \eqref{eq:PIDE_ex} is obtained by applying the inverse Fourier transform to \eqref{eq:fourier_soln}.

To implement this procedure, we first create a grid in the state space variable $x$ and in the frequency domain, $\zeta$. One then computes $\hat \omega(\cdot,\zeta)$ via \eqref{eq:fourier_soln} on this grid to approximate the solution to the PIDE \eqref{eq:fourier_PDE}. Applying the inverse Fourier transform to $\hat \omega(\cdot,\zeta)$ we obtain $\omega(\cdot,x)$, the solution to \eqref{eq:PIDE_ex}, on the state space grid. For detail on grid selection, we refer to \textcite{jackson2008fourier}.

\subsection{Simulating Sample Paths}\label{sec:simulation}
\begin{algorithm}[t!]
\setstretch{1.2}
\KwInput{constraint functions $f_i(x)$, constants $c_i$ for $i \in [n]$, intensity $\kappa$, severity distribution $G$, terminal time $T$, time step size $\Delta t$}\vspace{0.5em}

Compute $\eta^*_i$, $i \in [n]$ by using the FST method to solve \\[0.2em]
 \algindent $0 =  \E \left[ \exp \left( -\sum_{j=1}^n\eta^*_j\,f_j(X_T) \right)\left( f_i(X_T)-c_i \right)\right]$ \;

Set $X_0 = 0$\;
\For{$t = 0$ \KwTo $T$ \KwBy $\Delta t$}{
    Compute $\kappa^*(t,X_t)$ using \cref{alg:1}  with $\boldeta = \boldeta^*$\;
    Simulate whether a jump occurs using intensity rate $\kappa^*(t,X_t)$ \;
    If a jump occurs, sample from $G^*(t,X_t,dy)$ obtained in \cref{alg:1} with $\boldeta = \boldeta^*$\;
    Update $X_t$ \;
}
\KwOutput{$X$}
\caption{Simulating paths of $X_t$ under $\Q^*$}
\label{alg:2}
\end{algorithm}
Next, we illustrate how the FST method can be used to calculate the stressed intensity, the stressed severity distributions, and how to simulate sample paths under a stressed measure. For simplicity, we state the algorithm for a one-dimensional compound Poisson process $X$. We use the FST method to approximate (conditional) $\P$-expectations of functions applied to $X_T$. Recall that both the optimal Girsanov kernel $h^*(t,x,y)$ and the optimal Lagrange multipliers $\boldeta^*$ are computed by taking $\P$-expectations. Moreover, since $X_T$ is under $\P$ a compound Poisson random variable, its characteristic function (under $\P$) is given by the L\'{e}vy-Khintchine formula, i.e., $\Psi(\zeta) = \kappa  \left( \varphi_\xi(\zeta) -1 \right)$, where $\varphi_\xi$ is the characteristic function of the severity random variable $\xi$ under $\P$.

\cref{alg:2} describes the process of simulating under a stressed measure $\Q^*$. If the equations for the Lagrange multiplier are linear, a single application of FST will be sufficient. If the equations are however non-linear, we use a non-linear solver where the FST method is used to compute the equation for each iteration. To simulate paths of $X$ under the stressed measure $\Q^*$, we compute $\kappa^*(t,X_t)$ and $G^*(t,X_t)$ for a grid of times $t$, and simulate forward at each time step with this intensity and severity distribution.

The method to compute the intensity $\kappa^*(t,x)$ and severity distribution $G^*(t,x)$ is described in \cref{alg:1}. For this we first require the Girsanov kernel $h^*(t,x,y)$ which is obtain by computing
$\omega^*(t,x+y) = \E_{t,x+y}\left[\exp \left( -\sum_{i=1}^n\eta^*_i\,f_i(X_T) \right)\right]$ and $\omega^*(t,x)=\E_{t,x}\left[\exp \left( -\sum_{i=1}^n\eta^*_i\,f_i(X_T) \right)\right]$ using FST, and then taking their ratio. The intensity process $\kappa^*(t,x)$ is calculated via numerical integration with respect to $G(dy)$. Finally, the severity distribution $G^*(t,x)$ is obtained by re-sampling.

\begin{algorithm}[tbh!]
\setstretch{1.2}
\KwInput{time $t$, state variable $x$, Lagrange multiplier $\boldeta$, constraint functions $f_i(x)$, $i\in[n]$, intensity $\kappa$, severity distribution $G$, grid for numerical integration $y\_grid$, number of draws $N$}\vspace{0.5em}

\For{$y$ in $y\_grid$}{
    Compute $\omega^\boldeta(t,x+y) = \E_{t,x+y}\left[\exp \left( -\sum_{i=1}^n\eta_i\,f_i(X_T) \right)\right]$ and \\
     \algindent \algindent \quad 
     $\; \omega^\boldeta(t,x)=\E_{t,x}\left[\exp \left( -\sum_{i=1}^n\eta_i\,f_i(X_T) \right)\right]$ using the FST method \;
    Set $h^\boldeta(t,x,y) = \omega^\boldeta(t,x+y)/ \omega^\boldeta(t,x)$ \;
}
Approximate $\kappa^\boldeta(t,x) = \kappa \int_\R h^\boldeta(t,x,y) G(dy)$ using numerical integration \;
Draw $N$ times from $G(dy)$ \;
Compute weights by interpolating $h^\boldeta(t,x,y)$ in $y$ at these draws and normalizing by $\kappa^\boldeta(t,x)$\;
Obtain draws from $G^\boldeta(t,x,dy)$ by re-sampling the draws from $G(dy)$ with the computed weights \;
\KwOutput{$\kappa^\boldeta(t,x)$,  $G^\boldeta(t,x)$}
\caption{Approximating $\kappa^\boldeta(t,x)$ and $G^\boldeta(t,x)$ for fixed $t,x$}
\label{alg:1}
\end{algorithm}

To implement the case when $\boldsymbol{X}$ is bivariate, we follow the same procedure as in the univariate case, except that the jump distribution is now bivariate. Using the \texttt{copula} package \citep{hofert2018elements}, we set the desired marginal jump distributions and copula. \cref{alg:1} is modified so that $N$ draws are taken from this joint distribution but the weights are computed using only the draws corresponding to the stressed (first) marginal distribution. The bivariate distribution is then re-sampled with these weights. We then simulate forward using \cref{alg:2} with the two distorted marginal severity distributions and the shared distorted intensity process. 

\section{Conclusion}
In a dynamic setting, we consider a risk management framework based on the concept of reverse stress testing where the reference model is a compound Poisson process, $X$, over a finite time horizon $[0,T$]. We solve the optimization problem where we seek the probability measure on the path-space of stochastic processes which has minimal KL-divergence from the reference measure and fulfills constraints which can be written as expected values of functions applied to the process at terminal time. We show that this solution is unique, and refer to it as the stressed measure. We characterize the Radon-Nikodym derivative of the stressed measure, and derive the dynamics of $X$ under it. We prove that under the stressed measure, $X$ is a generalized version of a compound Poisson process, where both the intensity and the severity distribution depend on time and state. For general constraints, we provide a simulation algorithm which allows to simulate sample paths of $X$ under the stressed measure.

Of particular interest are constraints on VaR and CVaR risk measures, for which we analyze in detail the dynamics of $X$ under the stressed measures. 
In the multivariate setting and for constraints at time points earlier than $T$, we investigate how a stress on one component of the process alters the dependence between the components of the process and otherwise affects the components other than the one stressed. We find that for most dependence structures of the severity distribution, the stress increases the dependence between the portfolio components due to the shared intensity process, as well as changes in the the tail dependence of the process components. Finally, we consider a ``what-if'' scenario where we examine how severe a stress on a portfolio component needs to be to breach a risk threshold of the aggregate portfolio at a later time point. We find that smaller stresses are needed to the upper tail of a sub-portfolio to breach a VaR risk threshold as compared to stresses on the median or lower tails.

\section*{Acknowledgements}
EK is supported by an NSERC Canada Graduate Scholarship-Doctoral. SJ and SP would like to acknowledge support from the Natural Sciences and Engineering Research Council of Canada (grants RGPIN-2018-05705, RGPAS-2018-522715, and DGECR-2020-00333, RGPIN-2020-04289). SP also acknowledges the support from the Canadian Statistical Sciences Institute (CANSSI).

\printbibliography

\end{document}